\documentclass[a4paper,english,final]{IEEEtran}
\usepackage[latin9]{luainputenc}
\usepackage{array}
\usepackage{float}
\usepackage{amsmath}
\usepackage{amsthm}
\usepackage{amssymb}
\usepackage{graphicx}
\usepackage{setspace}

\makeatletter

\@ifundefined{pageheight}{\let\pageheight\pdfpageheight}{}
\@ifundefined{pagewidth}{\let\pagewidth\pdfpagewidth}{}
\pageheight\paperheight
\pagewidth\paperwidth

\providecommand{\tabularnewline}{\\}
\floatstyle{ruled}
\newfloat{algorithm}{tbp}{loa}
\providecommand{\algorithmname}{Algorithm}
\floatname{algorithm}{\protect\algorithmname}

  \theoremstyle{plain}
  \newtheorem{lem}{\protect\lemmaname}
  \theoremstyle{plain}
  \newtheorem{prop}{\protect\propositionname}

\usepackage{amsmath,amssymb}
\usepackage{color}  
\usepackage{graphicx,psfrag,cite,subfigure}

\usepackage{cite}

\author{
Omid~Esrafilian,\textit{ Student Member, IEEE}, Rajeev~Gangula, \textit{Member, IEEE}, and David~Gesbert, \textit{Fellow,~IEEE}

\thanks{This work was supported by the ERC under the European Union Horizon 2020 research and innovation program (Agreement no. 670896).}
\thanks{The authors are with the Department of Communication Systems, EURECOM, Sophia-Antipolis, France (email:\{esrafili, gangula, gesbert\}@eurecom.fr).}
}


\usepackage[acronym]{glossaries}
\newcommand{\newac}{\newacronym}
\newcommand{\ac}{\gls}

\makeglossaries

\newac{los}{\text{LoS}}{line-of-sight}
\newac{nlos}{\text{NLoS}}{non-line-of-sight}
\newac{uav}{UAV}{unmanned aerial vehicle}
\newac{speb}{SPEB}{square position error bound}
\newac[plural=EFIMs,firstplural=Fisher information matrices (EFIMs)]{efim}{EFIM}{Fisher information matrix}
\newac{ne}{NE}{Nash equilibrium}
\newac{mse}{MSE}{mean squared error}
\newac{toa}{TOA}{time-of-arrival}
\newac{snr}{SNR}{signal-to-noise ratio}
\newac{lan}{LAN}{local area network}
\newac{psd}{PSD}{positive semidefinite}
\newac{pd}{PD}{positive definite}
\newac{wrt}{w.r.t.}{with respect to}
\newac{lhs}{L.H.S.}{left hand side}
\newac{wp1}{w.p.1}{with probability 1}
\newac{kkt}{KKT}{Karush-Kuhn-Tucker}
\newac{wlog}{w.l.o.g.}{without loss of generality}
\newac{mle}{MLE}{maximum likelihood estimation}
\newac{gps}{GPS}{global positioning system}
\newac{rssi}{RSSI}{received signal strength}
\newac{mimo}{MIMO}{multiple-input multiple-output}
\newac{csi}{CSI}{channel state information}
\newac{fdd}{FDD}{frequency division duplexing}
\newac{ms}{MS}{mobile station}
\newac{bs}{BS}{base station}
\newac{d2d}{D2D}{device-to-device}
\newac{slnr}{SLNR}{signal-to-interference-leakage-and-noise-ratio}
\newac{ula}{ULA}{uniform linear antenna array}
\newac{pas}{PAS}{power angular spectrum}
\newac{mmse}{MMSE}{minimum mean square error}
\newac{zf}{ZF}{zero-forcing}
\newac{rzf}{RZF}{regularized zero-forcing}
\newac{as}{AS}{angular spread}
\newac{aod}{AOD}{angle of departure}
\newac{iid}{i.i.d.}{independent and identically distributed} 
\newac{sinr}{SINR}{signal-to-interference-and-noise ratio}
\newac{tdd}{TDD}{time-division duplex}
\newac{rvq}{RVQ}{random vector quantization}
\newac{rhs}{R.H.S.}{right hand side}
\newac{mrc}{MRC}{maximum ratio combining}
\newac{cdf}{CDF}{cumulative distribution function}
\newac{a.s.}{a.s.}{almost surely}
\newac{jsdm}{JSDM}{joint spatial division and multiplexing}
\newac{map}{MAP}{maximum a posteriori}
\newac{klt}{KLT}{Karhunen-Lo\`eve Transform}
\newac{lbe}{LBE}{link bargaining equilibrium}
\newac{se}{SE}{Stackelberg equilibrium}
\newac{pdf}{PDF}{probability density function}
\newac{em}{EM}{expectation-maximization}
\newac{knn}{KNN}{$k$-nearest neighbor}
\newac{svd}{SVD}{singular value decomposition}

\makeatother

\usepackage{babel}
  \providecommand{\lemmaname}{Lemma}
  \providecommand{\propositionname}{Proposition}

\begin{document}
%
%



\ifdefined\SINGLECOLUMN
	\onecolumn
	\setkeys{Gin}{width=0.5\columnwidth}
	\newcommand{\figfontsize}{\footnotesize} 
\else
	\setkeys{Gin}{width=1.0\columnwidth}
	\newcommand{\figfontsize}{\normalsize} 
\fi

\global\long\def\Tr{\text{T}}

\global\long\def\trace{\text{tr}}

\title{Learning to Communicate in UAV-aided Wireless Networks: Map-based
Approaches}
\maketitle
\begin{abstract}
We consider a scenario where an UAV-mounted flying base station is
providing data communication services to a number of radio nodes spread
over the ground. We focus on the problem of resource-constrained UAV
trajectory design with (i) optimal channel parameters learning and
(ii) optimal data throughput as key objectives, respectively. While
the problem of throughput optimized trajectories has been addressed
in prior works, the formulation of an optimized trajectory to efficiently
discover the propagation parameters has not yet been addressed. When
it comes to the communication phase, the advantage of this work comes
from the exploitation of a 3D city map. Unfortunately, the communication
trajectory design based on the raw map data leads to an intractable
optimization problem. To solve this issue, we introduce a map compression
method that allows us to tackle the problem with standard optimization
tools. The trajectory optimization is then combined with a node scheduling
algorithm. The advantages of the learning-optimized trajectory and
of the map compression method are illustrated in the context of intelligent
IoT data harvesting.
\end{abstract}


\section{Introduction}

\IEEEPARstart{T}{he} use of unmanned aerial vehicles (UAVs) also
known as drones as base stations (BSs) in future wireless communication
networks is currently gaining significant attention for its ability
to yield ultra-flexible deployments, in use cases ranging from disaster
recovery scenarios, coverage of flash-crowd events, and data harvesting
in IoT applications \cite{zenZhanLim_tuto,HayatYanmMuz,MozSaadBennNamDebb}.

Several new and fascinating issues arise from the study of flying
BSs in a wireless network. These can be broadly categorized into placement
and path planning problems. While placement problem deals with finding
flexible yet static locations of the UAV BSs, path planning involves
finding UAV trajectories. In both cases the aim is to optimize metrics
like throughput, network coverage, energy efficiency, etc., \cite{AlhKaLa,MozSaaBen,AlzEiLagYan,EsrGesbertWSA,LadPawHyoWen,ChenGesb,GanKerrEsrGes,WuQiYoRu,KleKabGir,ZenZhanLim,ZenZhang}.

When it comes to placement or trajectory design problems, most existing
solutions rely on simplified channel attenuation models which are
based on either (deterministically guaranteed) \ac{los} links \cite{LadPawHyoWen,ChenGesb,GanKerrEsrGes,WuQiYoRu,ZenZhang,ZenZhanLim,KleKabGir},
or predictive models for the probability of occurrence of a \ac{los}
link \cite{AlzEiLagYan,EsrGesbertWSA,AlhKaLa,MozSaaBen}. In the latter
approach, a \textit{global} statistical model predicts the \ac{los}
availability as a function of, e.g., UAV altitude and distance to
the user. The advantage of the global statistical \ac{los} model
lies in its simplicity for system analysis. However, it lacks actual
performance guarantees for either placement or trajectory design algorithms
when used in a real-life navigation scenario. The key reason for this
is that, the \textit{local} terrain topology may sharply differ from
the predictions drawn from statistical features. In order to circumvent
this problem, the embedding of actual 3D city map data in the UAV
placement algorithms has been recently proposed \cite{LadPawHyoWen,EsrGesbertWSA}.
Map-based approaches help providing a reliable prediction of \ac{los}
availability for any pair of UAV and ground node locations, hence
lead to improved performance guarantees. However, the gain comes at
the expense of computational and memory costs related to processing
of the map data. So far, map-based approaches have been investigated
mainly for static UAV placement \cite{LadPawHyoWen,ChenGesb,EsrGesbertWSA}.
In many scenarios, including IoT data harvesting, there is an interest
in flying along a path that brings the UAV-mounted BS closer to each
and every ground node. However, to the best of our knowledge, none
of the previous works have considered the crucial advantage of exploiting
3D map data in communication-oriented UAV trajectory design.

Another common assumption in previous works is that the channel model
parameters are assumed to be known when designing the UAV trajectory.
However, in reality these parameters need to be learned based on the
measurements collected from the ground users. As a result, an important
question arises: What is an efficient way of collecting radio measurements
by the UAV from the ground users in order to estimate the channel
model parameters?

In this work, we consider an instance of IoT data harvesting scenario
where an UAV flies over a city endowed with a number of scattered
ground nodes (e.g. radio-equipped sensors). We then formulate a resource-constrained
UAV trajectory design problem in order to optimize data throughput
from the ground nodes while capitalizing on 3D map data. Since the
data communication phase critically depends on the knowledge of radio
channel parameters, we also formulate a novel optimal trajectory design
problem from a \textit{parameter learning} point of view. Specifically,
our contributions are as follows:
\begin{itemize}
\item We formulate and solve \textit{a learning} trajectory optimization
problem in order to minimize the estimation error of the channel model
parameters. The devised trajectory allows the UAV to exploit the map
and quickly learn the propagation parameters.
\item Based on the learned parameters, we formulate a joint trajectory and
node scheduling problem which allows us to maximize the traffic communicated
from each node to the UAV in a fair manner. An iterative algorithm
is proposed to solve the optimization problem. It is shown that the
algorithm converges at least to a local optimal solution. While the
algorithm exploits the possibly rich map data, it does so via a \textit{map
compression} method which renders the trajectory optimization problem
differentiable and amenable to standard optimization tools, hence
mitigating a known drawback of map-based approaches. 
\end{itemize}
Note that, the proposed map-compression method allows us to smooth
out the map data while preserving the node location-dependent channel
behavior around that node. The gains brought by the exploitation of
3D map data are illustrated, in both channel parameters learning and
the communication problem, in the context of an urban IoT scenario.

This paper is organized as follows: Section \ref{sec:SysModel} introduces
the system model. In Section III, we formulate and solve the problem
of learning trajectory optimization. In Section IV, we formulate the
joint communication trajectory and node scheduling optimization problem
whose solution is provided in Section V. Numerical results are presented
in Section VI to validate the performance of the proposed algorithms.
Finally, Section VII concludes the paper with some perspectives.

\textbf{Notation: }Matrices are represented by uppercase bold letters,
vectors are represented by lowercase bold letters. The transpose of
matrix $\mathbf{A}$ is denoted by $\mathbf{A}^{\Tr}$.  The trace
and determinant of matrix $\mathbf{A}$ are denoted by $\text{tr}[\mathbf{A}]$
and $\text{det[}\mathbf{A}]$, respectively. The set of integers from
$m$ to $n$, $m<n$, is represented by $[m,\,n]$. The expectation
operator is denoted by $\text{E}[.]$.

\section{SYSTEM MODEL\label{sec:SysModel}}

A wireless communication system where an UAV-mounted flying BS serving
$K$ static ground level nodes (IoT sensors, radio terminals, etc.)
in an urban area is considered. The $k$-th ground node, $k\in\left[1,K\right]$,
is located at ${\bf u}_{k}=[x_{k},y_{k},0]^{\Tr}\in\mathbb{R}^{3}.$
By no means the ground level node assumption is restrictive, the proposed
algorithms in this work can in principle be applied to a scenario
where the nodes are located in 3D. The UAV's mission consists of a
learning phase which is of duration $T_{l}$ and then followed by
a communication phase of duration $T_{c}$. During the learning phase,
the UAV estimates the propagation parameters of the environment by
collecting the radio measurements from the ground users. These estimates
are then exploited in the communication phase to optimally serve the
ground nodes. Note that in this paper, we are treating the learning
and the communication phases that are separated in time. This allows
us to obtain optimal trajectories for both phases independently, i.e.,
if one is only interested in learning or communication scenario this
solution serves the purpose. The joint channel learning and communication
trajectory design is appealing yet a challenging problem and is left
for future work. Whether be it in learning or communication phase,
the time-varying coordinate of the UAV/drone is denoted by ${\bf v}(t)=[x(t),y(t),z(t)]^{\Tr}\in\mathbb{R}^{3}$,
where $z(t)$ represents the altitude of the drone.

For the ease of exposition, we assume that the time period $T_{l}$
and $T_{c}$ are discretized into $N_{l}$ and $N_{c}$ equal-time
slots, respectively. The time slots are chosen sufficiently small
such that the UAV's location, velocity, and channel gains can be considered
to remain constant in one slot. Hence, the UAV's position ${\bf v}(t)$
is approximated by a sequence

\begin{doublespace}
\noindent 
\begin{equation}
{\bf v}[n]=[x[n],y[n],z[n]]^{\Tr},~n\in[1,N_{l}]\ \mbox{(or) }n\in[1,N_{c}].
\end{equation}

\end{doublespace}

\noindent We assume that the ground nodes and the drone are equipped
with GPS receivers, hence the coordinates ${\bf u}_{k},\forall k$
and ${\bf v}[n],~n\in[1,N_{l}]\ \mbox{(or) }n\in[1,N_{c}]$ are known. 

\subsection{Channel Model}

In this section, we describe the channel model that is used for computing
the channel gains between the UAV and the ground users. The model
parameters are estimated in the learning phase from the collected
radio measurements. Classically, the channel gain between two radio
nodes which are separated by distance $d$ meters is modeled as\cite{Goldsmith,ChenYanGes}

\noindent 
\begin{equation}
\gamma_{s}=\frac{\beta_{s}}{d^{\alpha_{s}}}\times\xi_{s},\label{eq:CH_Model}
\end{equation}
where $\alpha_{s}$ is the path loss exponent, $\beta_{s}$ is the
average channel gain at the reference point $d=1$ meter, $\xi_{s}$
denotes the shadowing component, and finally $s\in\left\{ \ac{los},\text{NLoS}\right\} $
emphasizes the strong dependence of the propagation parameters on
\ac{los} or \ac{nlos} scenario. Note that (\ref{eq:CH_Model}) represents
the channel gain which is averaged over the small scale fading of
unit variance. The channel gain in dB can be written as
\begin{equation}
g_{s}=\ss{}_{s}-\alpha_{s}\varphi(d)+\eta_{s},\label{eq:CH_Model_dB}
\end{equation}
where $g_{s}=10\log_{10}\gamma_{s},\ss_{s}=10\log_{10}\beta_{s},\varphi(d)=10\log_{10}\left(d\right)$,
$\eta_{s}=10\log_{10}\xi_{s}$, and $\eta_{s}$ is modeled as a Gaussian
random variable with $\mathcal{N}(0,\sigma_{s}^{2})$.

\subsection{UAV Model}

During the mission, drone's position evolves according to

\begin{subequations}\label{eq:DroneDyn}

\noindent 
\begin{align}
{\bf v}[n+1] & ={\bf v}[n]+\left[\begin{array}{c}
\cos\left(\phi[n]\right)\cos\left(\psi[n]\right)\\
\sin\left(\phi[n]\right)\cos\left(\psi[n]\right)\\
\sin\left(\psi[n]\right)
\end{array}\right]\rho[n]\,,\label{eq:DroneDyn_a}\\
h_{min}\leq & z[n]\leq h_{max},\ \forall n\in[1,N_{l}-1]\ \mbox{(or) }n\in[1,N_{c}-1],\label{eq:Alt_const}
\end{align}

\end{subequations}

\noindent where in the $n$-th time slot, $0\le\rho[n]\le\rho_{max}$
represents the distance traveled by the drone, $0\le\phi[n]\le2\pi$
and $-\frac{\pi}{2}\le\psi[n]\le\frac{\pi}{2}$ represent the heading
and elevation angles, respectively. The maximum distance traveled
in a time slot is denoted by $\rho_{max}$ and it depends on the maximum
velocity. The constraint (\ref{eq:Alt_const}) reflects the fact that
the drone always flies at an altitude higher than $h_{min}$ and lower
than $h_{max}$, with $h_{min}$ being the height of the tallest building
in the city. 

\section{Learning Trajectory Design\label{sec:LearningTrj}}

In this section, our goal is to find the UAV trajectory, over which
the channel measurements are collected from the ground nodes, that
results in the minimum estimation error of the channel model parameters.
While the problem of learning the channel parameters from a pre-determined
measurement data set has been addressed in the prior literature \cite{ChenYanGes,ChenEsrGesMit},
the novelty of our work lies in the concept of optimizing the flight
trajectory \textit{itself} so as to accelerate the learning process.
The channel measurement collection and learning process are described
next.

\subsection{Measurement Collection and Channel Learning}

\noindent In the learning phase, the measurement harvesting is performed
over an UAV trajectory that starts at a base position ${\bf x}_{b}\in\mathbb{R}^{3}$
and ends at a terminal position ${\bf x}_{t}\in\mathbb{R}^{3}$. Mathematically,
\begin{equation}
{\bf v}[1]={\bf x}_{b},~{\bf v}[N_{l}]={\bf x}_{t}.\label{eq:LearnBoundCond}
\end{equation}

\noindent The base position is typically the take-off base for the
UAV while ${\bf x}_{t}$ can be selected in different ways, including
${\bf x}_{t}={\bf x}_{b}$ (loop) or as a location where the communication
services are to begin right after the learning phase. In the $n$-th
time interval, $n\in[1,N_{l}]$, the measurements collected from the
ground nodes can be written as

\noindent 
\[
{\bf g}_{s,n}=\left[g_{s,1},g_{s,2},\cdots,g_{s,\delta_{s,n}}\right]^{\Tr},
\]

\noindent where $g_{s,i}^{\text{ }}$ is the channel gain of the $i$-th
measurement, $i\in[1,\delta_{s,n}]$, and $\delta_{s,n}$ is the number
of measurements obtained for the propagation segment group $s\in\{\text{LoS},\text{NLoS}\}$.
For the \ac{los}/\ac{nlos} classification of the measurements, we
leverage the knowledge of a 3D city map \cite{EsrGes}. Based on such
map, we can predict \ac{los} (un)availability on any given UAV-ground
nodes link from a trivial geometry argument: For a given UAV position,
the ground node is considered in \ac{los} to the UAV if the straight
line passing through the UAV's and the ground node's position lies
higher than any buildings in between.

Using (\ref{eq:CH_Model_dB}), the $i$-th measurement can be modeled
as

\noindent 
\begin{equation}
g_{s,i}^{\text{ }}={\bf {\bf a}}_{s,i}^{\Tr}\,{\bf \boldsymbol{\omega}}_{s}+\eta_{s,i}^{\text{ }},\label{eq:regression-function-1-1-1}
\end{equation}

\noindent where ${\bf a}_{s,i}=[-\varphi(d_{i}),\,1]^{\Tr}$ with
$d_{i}$ being the distance between the drone and ground node in the
$i$-th measurement, ${\bf \boldsymbol{\omega}}_{s}=[\alpha_{s},\,\ss_{s}]^{\Tr}$
is the vector of channel parameters, and $\eta_{s,i}$ denotes the
shadowing component in the $i$-th measurement. The measurements collected
in the $n$-th interval can now be written as

\noindent 
\begin{equation}
\mathbf{g}_{s,n}={\bf A}_{s,n}\,{\bf \boldsymbol{\omega}}_{s}+\boldsymbol{\eta}_{s,n},\label{eq:EachSegMeas-1-1}
\end{equation}
where ${\bf A}_{s,n}=\left[{\bf a}_{s,1},\cdots,{\bf a}_{s,\delta_{s,n}}\right]^{\Tr},\,\boldsymbol{\eta}_{s,n}=[\eta_{s,1}^{\text{ }},\cdots,\eta_{s,\mathcal{\delta}_{s,n}}^{\text{ }}]^{\Tr}$.
Finally, we stack up the measurements gathered by the drone up to
time step $n$ as
\begin{equation}
{\bf \bar{\mathbf{g}}}_{s,n}={\bf \bar{A}}_{s,n}\,{\bf \boldsymbol{\omega}}_{s}+\boldsymbol{\bar{\eta}}_{s,n},\label{eq:TotalCHMeasModel}
\end{equation}

\noindent where ${\bf {\bf \bar{\mathbf{g}}}}_{s,n}=\left[\mathbf{g}_{s,1}^{\Tr},\cdots,\mathbf{g}_{s,n}^{\Tr}\right]^{\Tr},{\bf \bar{A}}_{s,n}=\left[{\bf A}_{s,1}^{\Tr},\cdots,{\bf A}_{s,n}^{\Tr}\right]^{\Tr},\,\text{and }\boldsymbol{\bar{\eta}}_{s,n}=\left[\boldsymbol{\eta}_{s,1}^{\Tr},\cdots,\boldsymbol{\eta}_{s,n}^{\Tr}\right]^{\Tr}$. 

Assuming that the measurements collected over a trajectory are independent,
the maximum likelihood estimation of ${\bf \boldsymbol{\omega}}_{s},s\in\{\text{LoS},\text{NLoS}\}$
based on the measurements collected up to time step $n$ is given
by \cite{ChenYanGes,EsrGes}

\noindent 
\begin{equation}
\hat{\boldsymbol{\omega}}_{s,n}=\left({\bf \bar{A}}_{s,n}^{\Tr}\,{\bf \bar{A}}_{s,n}\right)^{-1}{\bf \bar{A}}_{s,n}^{\Tr}\,{\bf \bar{\mathbf{g}}}_{s,n}.\label{eq:MLE_Sol}
\end{equation}
By substituting (\ref{eq:TotalCHMeasModel}) in (\ref{eq:MLE_Sol}),
we obtain
\begin{equation}
\hat{\boldsymbol{\omega}}_{s,n}-\boldsymbol{\omega}_{s}=\left({\bf \bar{A}}_{s,n}^{\Tr}\,{\bf \bar{A}}_{s,n}\right)^{-1}{\bf \bar{A}}_{s,n}^{\Tr}\boldsymbol{\bar{\eta}}_{s,n}.
\end{equation}

\noindent Since $\hat{\boldsymbol{\omega}}_{s,n}$ is unbiased, the
mean squared error of the estimated parameters can be obtained as
\cite{RogGiro} 
\begin{align}
\text{E}\left\Vert \hat{\boldsymbol{\omega}}_{s,n}-\boldsymbol{\omega}_{s}\right\Vert ^{2} & =\trace\left[Cov\left\{ \hat{\boldsymbol{\omega}}_{s,n}\right\} \right]\nonumber \\
 & =\sigma_{s}^{2}\,\trace\left[\left({\bf \bar{A}}_{s,n}^{\Tr}\,{\bf \bar{A}}_{s,n}\right)^{-1}\right].\label{eq:Param_MSE-1}
\end{align}

\noindent Let
\[
e_{s}[n]\triangleq\trace\left[\left({\bf \bar{A}}_{s,n}^{\Tr}\,{\bf \bar{A}}_{s,n}\right)^{-1}\right],
\]

\noindent and assuming that $\sigma_{\ac{nlos}}^{2}=\kappa\cdot\sigma_{\ac{los}}^{2},\,\kappa\geq1$
\cite{3GPP}, the total estimation error in both propagation segments
is given by
\begin{equation}
\sum_{s}\text{E}\left\Vert \hat{\boldsymbol{\omega}}_{s,n}-\boldsymbol{\omega}_{s}\right\Vert ^{2}=\sigma_{\ac{los}}^{2}\left(e_{\ac{los}}[N_{l}]+\kappa\,e_{\ac{nlos}}[N_{l}]\right).\label{eq:LearningTrjError}
\end{equation}
Note that a full rank ${\bf \bar{A}}_{s,n}$ is assumed in calculating
the error for both \ac{los} and \ac{nlos} categories over the course
of the trajectory. If there are not enough measurements in a particular
segment by the end of the trajectory, the estimation error is assigned
as infinity in that segment.

\subsection{Optimization Problem}

The optimal learning trajectory that minimizes the estimation error
can be formulated as

\noindent \begin{subequations}\label{eq:Learn_GenOpt}
\begin{align}
\min_{\Phi,\Psi,\mathcal{R}}\  & e_{\ac{los}}[N_{l}]+\kappa\,e_{\ac{nlos}}[N_{l}]\label{eq:GenLearPathOpt}\\
\mbox{s.t.\ } & \eqref{eq:DroneDyn},\,\eqref{eq:LearnBoundCond}
\end{align}

\end{subequations}

\noindent where $\Phi,\Psi,$ and $\mathcal{R}$ are defined as 
\begin{align*}
\Phi & =\left\{ 0\le\phi[n]<2\pi,\,n\in[1,N_{l}-1]\right\} ,\\
\Psi & =\left\{ -\frac{\pi}{2}\le\psi[n]\le\frac{\pi}{2},\,n\in[1,N_{l}-1]\right\} ,\\
\mathcal{R} & =\left\{ 0\le\rho[n]\le\rho_{max},\,n\in[1,N_{l}-1]\right\} .
\end{align*}

\noindent As the estimation error depends on the matrix ${\bf \bar{A}}_{s,N_{l}}$
which has a very complicated expression in terms of $\phi[n],\psi[n],$
and $\rho[n]$, it is hard to obtain an analytical solution for problem
\eqref{eq:Learn_GenOpt} in general. Therefor, we tackle \eqref{eq:Learn_GenOpt}
by discretizing the optimization variables and then employing dynamic
programming (DP) \cite{Kirk} to find the solution. To apply DP, the
estimation error $e_{s}[N_{l}]$ needs to be rewritten as follows
\begin{align}
e_{s}[N_{l}] & =\trace\left[\left(\left[\begin{array}{c}
{\bf \bar{A}}_{s,N_{l}-1}\\
{\bf A}_{s,N_{l}}
\end{array}\right]^{\Tr}\,\left[\begin{array}{c}
{\bf \bar{A}}_{s,N_{l}-1}\\
{\bf A}_{s,N_{l}}
\end{array}\right]\right)^{-1}\right]\nonumber \\
 & \overset{(a)}{=}e_{s}[N_{l}-1]-r_{s}[N_{l}]\nonumber \\
 & \overset{(b)}{=}e_{s}[1]-\sum_{n=2}^{N_{l}}r_{s}[n],
\end{align}

\noindent where we denote $r_{s}[n]$ as the amount of improvement
in the estimate within time slot $n$, and it is given by
\begin{equation}
r_{s}[n]=\trace\left[{\bf H}_{s,n}\,{\bf A}_{s,n}^{\Tr}\left({\bf I}+{\bf A}_{s,n}\,{\bf H}_{s,n}\,{\bf A}_{s,n}^{\Tr}\right)^{-1}\,{\bf A}_{s,n}\,{\bf H}_{s,n}\right],
\end{equation}

\noindent ${\bf H}_{s,n}\triangleq\left({\bf \bar{A}}_{s,n-1}^{\Tr}\,{\bf \bar{A}}_{s,n-1}\right)^{-1}$,
${\bf I}$ is the identity matrix, (a) follows from the matrix inversion
lemma, and (b) follows from the recursive relation. Now \eqref{eq:Learn_GenOpt}
can be reformulated as 

\noindent \begin{subequations}\label{eq:Learn_ReGenOpt}
\begin{align}
\min_{\Phi,\Psi,\mathcal{R}}\  & \sum_{n=1}^{N_{l}}\tilde{e}_{\ac{los}}[n]+\kappa\,\tilde{e}_{\ac{nlos}}[n]\label{eq:GenLearPathOpt-1}\\
\mbox{s.t.\ } & \eqref{eq:DroneDyn},\,\eqref{eq:LearnBoundCond}
\end{align}

\end{subequations}

\noindent where

\noindent 
\[
\tilde{e}_{s}[n]=\begin{cases}
e_{s}[1] & n=1\\
-r_{s}[n] & n\in[2,N_{l}]
\end{cases}.
\]

\subsection{Dynamic Programming\label{subsec:DiscreteApproxAndDP}}

To solve \eqref{eq:Learn_ReGenOpt} by DP, we constraint (and thus
approximate) the possible drone locations and the optimization variables
to a limited alphabet and then use Bellman's recursion to compute
the optimal discrete trajectory. We start by introducing some notations.

\noindent Let ${\bf v}[n],\,n\in[1,N_{l}]$ denotes the states and
$\boldsymbol{\pi}[n]=\left[\phi[n],\psi[n],\rho[n]\right]^{\Tr}$
represents the input action at time $n\in[1,N_{l}-1]$ such that

\noindent 
\begin{align}
\phi[n] & \in\left\{ 0,\frac{\pi}{4},\frac{\pi}{2},\frac{3\pi}{4},\pi,\frac{5\pi}{4},\frac{3\pi}{2},\frac{7\pi}{4}\right\} ,\nonumber \\
\psi[n] & \in\left\{ -\frac{\pi}{2},-\frac{\pi}{4},0,\frac{\pi}{4},\frac{\pi}{2}\right\} ,\nonumber \\
\rho[n] & \in\left\{ 0,a_{h},a_{v},a_{h}\sqrt{2},\sqrt{a_{h}^{2}+a_{v}^{2}},\sqrt{2a_{h}^{2}+a_{v}^{2}}\right\} ,\label{eq:InputAlpbet}
\end{align}

\noindent where $a_{h}$ and $a_{v}$ denote the discretization unit
used in discretizing the city map into a 3D grid (hereafter termed
as path graph) of admissible drone locations. Depending on the action
${\bf \boldsymbol{\pi}}[n]$ in ${\bf v}[n]$, the state ${\bf v}[n+1]$
can be computed by using \eqref{eq:DroneDyn} and (\ref{eq:InputAlpbet}).
In Fig. \ref{fig:Path_Graph}, a part of the path graph, arbitrary
base position ${\bf x}_{b},$ and terminal position ${\bf x}_{t}$
are illustrated. The vertices and the edges of the path graph can,
respectively, be interpreted as the admissible states and input actions
in each time slot.

In order not to exceed the flight time constraint $T_{l}$, $N_{l}$
can be selected as \footnote{Note that this is a conservative choice. In practice, $N_{l}$ could
be slightly higher given the UAV may use some of the short edges.}

\noindent 
\[
N_{l}=\left\lfloor \frac{T_{l}}{T_{e}}\right\rfloor ,
\]

\begin{spacing}{0.92}
\noindent where $\left\lfloor .\right\rfloor $ denotes the floor
function and $T_{e}=\frac{\sqrt{2a_{h}^{2}+a_{v}^{2}}}{v_{max}}$
is the minimum required time for taking the longest edge between two
adjacent vertices in the path graph while the drone moves with maximum
speed $v_{max}$.
\end{spacing}

DP in a forward manner is now used to solve for \eqref{eq:Learn_ReGenOpt}
by taking into account the finite alphabet constraint (\ref{eq:InputAlpbet}).
Thus, by reformulating (\ref{eq:GenLearPathOpt-1}) we can associate
with our problem the performance index
\begin{equation}
J_{i}({\bf v}[i])=\Omega({\bf v}[1])+\sum_{n=2}^{i}L[n],
\end{equation}

\noindent where $[2,i]$ is the time interval of interest and $L[n]=\tilde{e}_{\ac{los}}[n]+\kappa\,\tilde{e}_{\ac{nlos}}[n]$.
$\Omega({\bf v}[1])$ stands for the initial cost and given by
\[
\Omega({\bf v}[1])=\begin{cases}
L[1] & {\bf v}[1]={\bf x}_{b}\\
\infty & \mbox{otherwise }
\end{cases}.
\]

According to Bellman's equation, the optimal cost up to time $n+1$
is equal to

\begin{subequations}\label{eq:DP_OPT}

\noindent 
\begin{align}
J_{n+1}^{*}({\bf v}[n+1]) & =\min_{{\bf \boldsymbol{\pi}}[n]}\left\{ L[n+1]+J_{n}^{*}({\bf v}[n])\right\} ,\,n\in[1,N_{l}-1],\label{eq:BellmanEq}\\
J_{1}^{*}({\bf v}[1]) & =\Omega({\bf v}[1]),
\end{align}

\end{subequations}

\noindent where ${\bf \boldsymbol{\pi}}[n]$ is the input action vector.
Thus, the optimal input action $\boldsymbol{\pi}^{*}[n]$ at time
$n$ is the one that achieves the minimum in (\ref{eq:BellmanEq}).
Finally, the optimal policy (trajectory) can be found by solving \eqref{eq:DP_OPT}
for all $n\in[1,N_{l}-1]$ and by choosing ${\bf v}[N_{l}]={\bf x}_{t}$.
Note that the number of computations required to find the optimal
trajectory is given by \cite{Kirk} 
\[
\mathcal{V}\cdot\Pi\cdot N_{l},
\]

\noindent where $\mathcal{V}$ is the number of admissible states
(i.e. the number of vertices in the path graph), and $\Pi$ is the
number of quantized admissible input actions. 

Note that the error $L[n]$ only depends on the UAV location through
its distance from the ground users and the \ac{los}/\ac{nlos} status.
Since we have the knowledge of the 3D map and the ground nodes' locations,
\eqref{eq:DP_OPT} can be solved offline without collecting any measurements.
Once the optimal trajectory is calculated, UAV follows this trajectory
to collect the measurements and then estimates the channel parameters.

\begin{spacing}{0.5}
\begin{figure}[t]
\begin{centering}
\includegraphics[width=5cm]{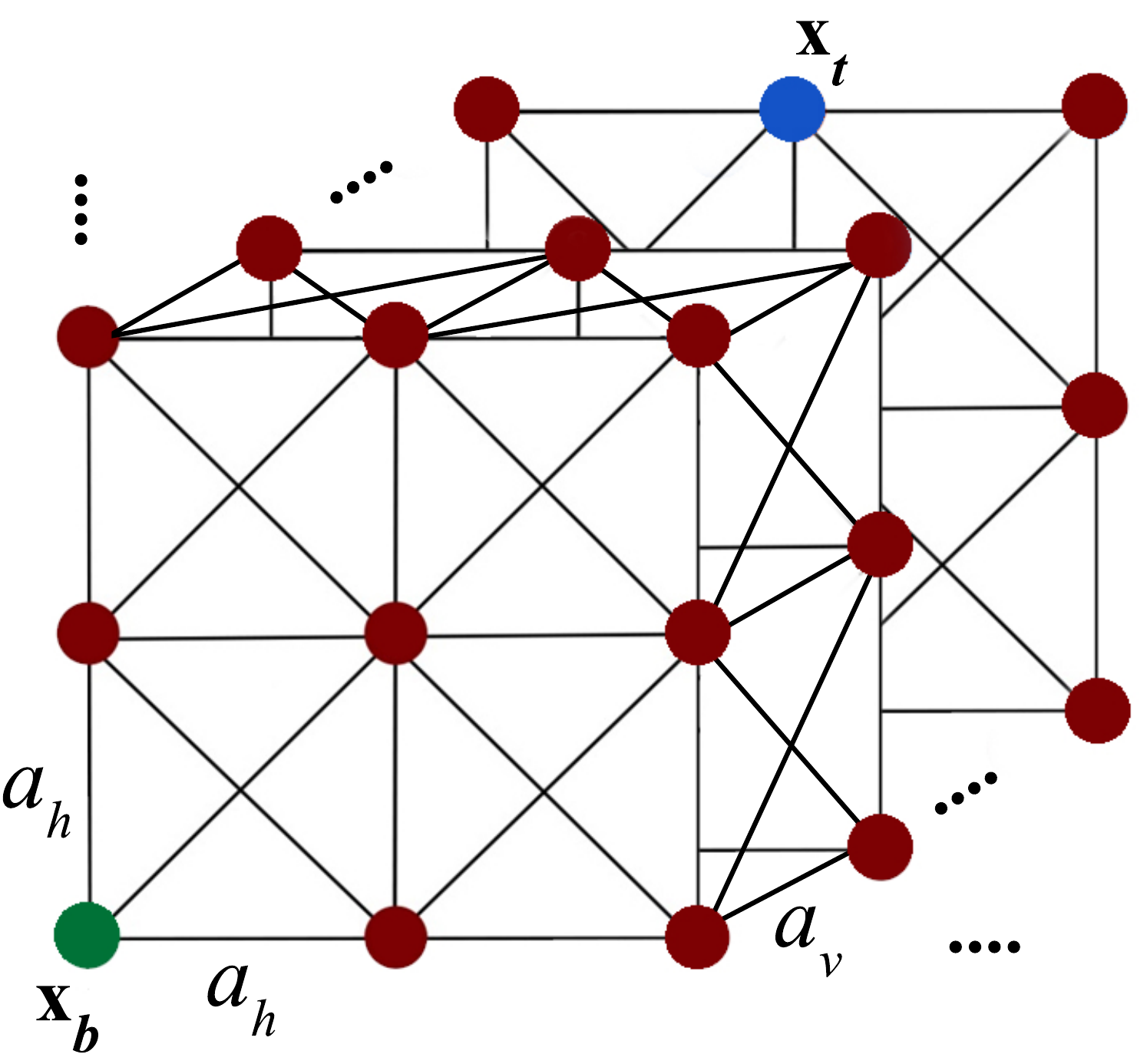}
\par\end{centering}
\caption{A fragment of the 3D path graph, arbitrary base position ${\bf x}_{b},$
and terminal position ${\bf x}_{t}$ .\label{fig:Path_Graph}}
\end{figure}

\end{spacing}

\section{Communication trajectory Optimization\label{sec:CommTrj}}

Based on the acquired knowledge of the channel parameters from the
learning phase, we are now concerned with the design of a \textit{communication
trajectory} in an uplink IoT data harvesting scenario. For the ease
of exposition, the communication phase is designed based on the perfect
channel model parameter estimates, while the impact of imperfect estimation
is addressed in Section \ref{subsec:Sim_Comm_TRJ}.

\subsection{Communication System Model}

We assume that the ground nodes are served by the drone in a time-division
multiple access (TDMA) manner. Let $q_{k}[n]$ denotes the scheduling
variable, then the TDMA constraints can be written as

\noindent 
\begin{align}
 & \sum_{k=1}^{K}q_{k}[n]\le1\,,\,n\in[1,N_{c}],\label{eq:TDD_Cons}\\
 & q_{k}[n]\in\left\{ 0,1\right\} ,\,n\in[1,N_{c}],\,k\in[1,K],\label{eq:Sch_Var_Cons}
\end{align}

\noindent where $q_{k}[n]=1$ indicates that the node $k$ is scheduled
in time slot $n$. For the scheduled node, the average throughput
is given by

\noindent 
\begin{equation}
C_{k}[n]=\text{E}\,\log_{2}\left(1+\frac{P\gamma_{k}[n]}{\sigma^{2}}\right),
\end{equation}

\noindent where $\gamma_{k}[n]$ is the channel gain between the $k$-th
node and the UAV at time step $n$, $P$ denotes the up-link transmission
power of the ground node, and the additive white Gaussian noise power
at the receiver is denoted by $\sigma^{2}$. Hence, the average achievable
throughput of the $k$-th ground node over the course of the communication
trajectory is given by

\noindent 
\begin{equation}
C_{k}=\frac{1}{N_{c}}\sum_{n=1}^{N_{c}}q_{k}[n]C_{k}[n].
\end{equation}

\subsection{Joint Scheduling and Trajectory Optimization }

We consider the problem of efficient data harvesting where data originates
from ground nodes and efficiency is meant in a max-min sense across
the nodes. The problem of maximizing the minimum average throughput
among all ground nodes by jointly optimizing node scheduling and UAV's
trajectory can be formulated as

\begin{subequations}\label{eq:GeneralOptProb}

\noindent 
\begin{align}
\max_{\mathcal{X},z,Q}\  & \underset{k\in[1,K]}{\min}\,C_{k}\label{eq:GeneralOptimizationProb}\\
\mbox{s.t.\ } & \left\Vert {\bf v}[n]-{\bf v}[n-1]\right\Vert \le\rho_{max}\ ,n\in[2,N_{c}],\label{eq:GenProb_Dmax}\\
 & {\bf v}[1]={\bf v}[N_{c}],\label{eq:GenProb_loopconst}\\
 & \eqref{eq:Alt_const},\,\eqref{eq:TDD_Cons},\,\eqref{eq:Sch_Var_Cons},
\end{align}

\end{subequations}

\noindent where $Q=\{q_{k}[n],\,\forall n,\,\forall k\}$ is the set
of scheduling variables, and $\mathcal{X}=\{(x[n],y[n]),\:\forall n\}$
denotes the discretized trajectory set of length $N_{c}$ in 2D. We
assume that the drone flies at a fixed altitude $z[n]=z,\,\forall n$.
The maximum speed constraint of the UAV is reflected in (\ref{eq:GenProb_Dmax}),
where $\rho_{max}=v_{max}T_{c}/N_{c}$. (\ref{eq:GenProb_loopconst})
implies a possible loop trajectory constraint\footnote{This is by no means a restriction, the starting and the terminal points
of the trajectory can be any arbitrary locations.}.

Problem \eqref{eq:GeneralOptProb} is challenging to solve due to
the following issues:
\begin{itemize}
\item The scheduling variables $q_{k}[n]$ are binary and include integer
constraints.
\item The objective function (\ref{eq:GeneralOptimizationProb}) is a non-convex
with respect to the drone trajectory variables. 
\item Since the 3D city map, node locations, and UAV location at time $n$
are known, then in theory the \ac{los} or \ac{nlos} status of the
link can be finely predicted and, hence, the link gain $\gamma_{k}[n]$
can be computed from (\ref{eq:CH_Model}) up to the random shadowing.
Unfortunately, such a direct exploitation of the rich \textit{raw}
map data leads to a highly non-differentiable problem in \eqref{eq:GeneralOptProb}.
\end{itemize}
\noindent We overcome these difficulties by approximation using the
same framework as \cite{WuQiYoRu} by employing the block-coordinate
descent \cite{XuYin} and sequential convex programming \cite{DinDie}
techniques. However, the key difference is that we optimize the drone's
altitude, and also exploit the 3D city map by introducing a \textit{statistical
map compression} approach that enables us to take into account the
\ac{los} and \ac{nlos} predictions.

\subsection{\ac{los} Probability Model Using Map Compression }

Statistical map compression approach relies on converting 3D map data
to build a reliable node location dependent \ac{los} probability
model. The \ac{los} probability for the link between the drone located
at altitude $z$ and the $k$-th ground node in the $n$-th time slot
is given by

\noindent 
\begin{equation}
p_{k}[n]=\frac{1}{1+\exp(-a_{k}\theta_{k}[n]+b_{k})},\label{eq:LoS_Prob}
\end{equation}

\noindent where $\theta_{k}[n]=\arctan(z/r_{k}[n])$ denotes the elevation
angle and $r_{k}[n]$ is the ground projected distance between the
drone and the $k$-th node located at ${\bf u}_{k}$ in the time slot
$n$, and $\{a_{k},b_{k}\}$ are the model coefficients.

The \ac{los} probability model coefficients $\{a_{k},b_{k}\}$ are
learned (i.e. by utilizing logistic regression method\cite{RogGiro})
by using a training data set formed by a set of tentative UAV locations
around the $k$-th ground node along with the true \ac{los}/\ac{nlos}
label obtained from the 3D map. Interestingly, the model in (\ref{eq:LoS_Prob})
can be seen as a localized extension of the classical (global) \ac{los}
probability model used in \cite{AlhKaLa,MozSaaBen}. The key difference
lies in the fact that, a \textit{local} \ac{los} probability model
will give performance guarantees which a global model cannot.

Using (\ref{eq:LoS_Prob}), the average channel gain of the link between
the drone and the $k$-th ground node in the $n$-th time slot is

\noindent {\small{}
\begin{align}
\text{E}[\gamma_{k}[n]]= & \left(\frac{d_{k}^{(A-1)\alpha_{\ac{los}}}-B}{1+\exp(-a_{k}\theta_{k}+b_{k})}+B\right)\frac{\beta_{\ac{los}}}{d_{k}^{\alpha_{\ac{nlos}}}},\label{eq:AverageChGain}
\end{align}
}{\small \par}

\noindent where $B=\frac{\beta_{\ac{nlos}}}{\beta_{\ac{los}}},\,A=\frac{\alpha_{\ac{nlos}}}{\alpha_{\ac{los}}}\ge1$,
and $d_{k}[n]=\sqrt{z^{2}+r_{k}^{2}[n]}$ is the distance between
the $k$-th ground node and the drone. The details of the proof are
given in Appendix \ref{subsec:ProofOfAvgChGain}.

\section{proposed solution for communication trajectory optimization}

\noindent In this section, we first approximate the original optimization
problem in \eqref{eq:GeneralOptProb} to a \textit{map-compressed}
problem and then present an iterative algorithm using block coordinate
descent for solving it. Using (\ref{eq:AverageChGain}) and the Jensen's
inequality, the average throughput upper-bound then can be written
as

\noindent 
\begin{equation}
C_{k}^{\text{\text{up}}}=\frac{1}{N_{c}}\sum_{n=1}^{N_{c}}q_{k}[n]C_{k}^{\text{\text{up}}}[n]\ ,\,k\in[1,K],
\end{equation}

\noindent where

\noindent 
\begin{equation}
C_{k}^{\text{\text{up}}}[n]=\log_{2}\left(1+\frac{P\,\text{E}[\gamma_{k}[n]]}{\sigma^{2}}\right).\label{eq:Cap_upperBound}
\end{equation}

We then approximate the original problem in \eqref{eq:GeneralOptProb}
into the following \textit{map-compressed} problem:

\begin{subequations}\label{eq:RelaxedProb}

\noindent 
\begin{align}
\max_{\mathcal{X},z,Q,\mu} & \ \mu\label{eq:RelaxedAveProb}\\
\mbox{s.t.\ } & C_{k}^{\text{\text{up}}}\ge\mu\ ,\,\forall k,\label{eq:RelaxedAveProb_UserCap}\\
 & 0\le q_{k}[n]\le1\ ,\,\forall k,\,\forall n,\label{eq:Relaxed_Sche_Var}\\
 & \eqref{eq:GenProb_Dmax},\,\eqref{eq:GenProb_loopconst},\,\eqref{eq:Alt_const},\,\eqref{eq:TDD_Cons},\label{eq:Relaxed_Prob_Sch-1}
\end{align}

\end{subequations}

\noindent where the constraints in (\ref{eq:Relaxed_Sche_Var}) represent
the relaxation of the binary scheduling variable into continuous variables.
Moreover, map compression allows us to circumvent the non-differentiability
aspect of the original problem \eqref{eq:GeneralOptProb} by compressing
the 3D map information into a probabilistic LoS model. However, \eqref{eq:RelaxedProb}
is still difficult to solve since it is a joint scheduling and path
planning problem and is not convex. To make this problem more tractable,
we split it up into three optimization sub-problems and then classically
iterate between them to converge to a final solution. Note that, the
iteration index of the proposed algorithm is denoted by ``$j$''.

\subsection{Scheduling }

For a given UAV planar trajectory $\mathcal{X}$ and altitude $z$,
the ground node scheduling can be optimized as

\begin{subequations}\label{eq:ScheProb}

\noindent 
\begin{align}
\max_{Q,\mu}\  & \ \mu\\
\mbox{s.t.\ } & \ C_{k}^{\text{\text{up}}}\ge\mu\ ,\,\forall k,\\
 & \eqref{eq:TDD_Cons},\,\eqref{eq:Relaxed_Sche_Var}.\nonumber 
\end{align}

\end{subequations}

\noindent This problem is a standard Linear Program (LP) and can be
solved by using any optimization tools such as CVX \cite{CVX}.

\subsection{Optimal Horizontal UAV Trajectory\label{subsec:PathPlanningSec} }

For a given scheduling decision $Q$, and drone's altitude $z$, we
now aim to find the optimal planar trajectory by solving

\begin{subequations}\label{eq:AveOptProb}

\noindent 
\begin{align}
\max_{\mathcal{X},\,\mu}\  & \mu\label{eq:AveOptProb_obj}\\
\mbox{s.t.\ } & C_{k}^{\text{\text{up}}}\ge\mu\ ,\,\forall k,\label{eq:AveOptProb_UserCap}\\
 & \eqref{eq:GenProb_Dmax},\,\eqref{eq:GenProb_loopconst}.
\end{align}

\end{subequations}

\noindent The optimization problem \eqref{eq:AveOptProb} is not convex,
since the constraint (\ref{eq:AveOptProb_UserCap}) is neither convex
nor concave. In general, there is no efficient method to obtain the
optimal solution. Therefore, we adopt sequential convex programming
technique for solving \eqref{eq:AveOptProb}. To this end, the following
results are helpful.
\begin{lem}
\label{lem:lem1}The function $h(x,y)\triangleq\text{\ensuremath{\log}}\left(1+f\left(x\right)g\left(y\right)\right)$
is convex if $\hat{h}(x,y)\triangleq\text{\ensuremath{\log}}\left(f\left(x\right)g\left(y\right)\right)$
is convex and $f(x)>0,\,\text{and }g(y)>0$.
\end{lem}
\begin{proof}
See Appendix \ref{subsec:AppendA}.
\end{proof}
\begin{prop}
For any constant $\tau,\lambda>0$, the function $c(x,y,d)\triangleq\log\left(1+\left[(\frac{1}{1+x})(\frac{1}{y})+\tau\right]\frac{1}{d^{\lambda}}\right)$
is convex.\label{lem:Proposition1}
\end{prop}
\begin{proof}
See Appendix \ref{subsec:Apx_Prop1}.
\end{proof}
By defining the auxiliary variables $f_{k}[n],\,w_{k}[n]$, $l_{k}[n]$,
and $\theta_{k}[n]$, we can rewrite \eqref{eq:AveOptProb} as follows

\begin{subequations}\label{eq:EqualOptProb}

\noindent 
\begin{align}
\max_{\mathbb{V},\mathcal{X},\mu} & \ \mu\label{eq:EqualOptProb_obj}\\
\text{s.t.}\mbox{\ } & \frac{1}{N_{c}}\sum_{n=1}^{N_{c}}c_{k}^{\text{ }}\left(f_{k}[n],w_{k}[n],l_{k}[n]\right)\ge\mu\ ,\,\forall k,\label{eq:EqualOptProb_UserCap}\\
 & w_{k}[n]=\left(\left(z^{2}+l_{k}[n]\right)^{(A-1)\alpha_{\ac{los}}/2}-B\right)^{-1}\ ,\,\forall k,\,\forall n,\label{eq:EqualOptProb_gm}\\
 & f_{k}[n]=\exp\left(-a_{k}\theta_{k}[n]+b_{k}\right)\ ,\,\forall k,\,\forall n,\label{eq:EqualOptProb_F}\\
 & l_{k}[n]=r_{k}^{2}[n]\ ,\,\forall k,\,\forall n,\\
 & \theta_{k}[n]=\arctan\left(z/\sqrt{l_{k}[n]}\right)\ ,\,\forall k,\,\forall n,\label{eq:EqualOptProb_theta-1}\\
 & f_{k}[n],w_{k}[n],l_{k}[n],\theta_{k}[n]\ge0\ ,\,\forall k,\,\forall n,\label{eq:EqualOptProb_Vge0}\\
 & \eqref{eq:GenProb_Dmax},\,\eqref{eq:GenProb_loopconst},
\end{align}

\end{subequations}

\noindent where $\mathbb{V}=\{f_{k}[n],w_{k}[n],l_{k}[n],\theta_{k}[n]\,|\,\forall k,\,\forall n\}$
consists of all the auxiliary variables and
\[
c_{k}^{\text{ }}\left(f_{k}[n],w_{k}[n],l_{k}[n]\right)\triangleq\qquad\qquad\qquad\qquad\qquad\qquad\qquad
\]

\noindent 
\begin{equation}
\log_{2}\left(1+\left(\frac{1}{w_{k}[n](1+f_{k}[n])}+B\right)\frac{P\,\beta_{\ac{los}}}{\sigma^{2}\left(z^{2}+l_{k}[n]\right)^{\alpha_{\ac{nlos}}/2}}\right).\label{eq:CmFmDm}
\end{equation}

\noindent Using Proposition \ref{lem:Proposition1}, it can be easily
seen that (\ref{eq:CmFmDm}) is a convex function of variables $f_{k}[n],w_{k}[n],\text{ and }l_{k}[n]$.
In constraint (\ref{eq:EqualOptProb_gm}), $w_{k}[n]$ can be convex
or concave function depending on the value of $B$. However, in our
case it is always convex since $z\ge h_{min}$, in a realistic scenario
$\left(z^{2}+l_{k}[n]\right)^{(A-1)\alpha_{\ac{los}}/2}\gg B$. Moreover,
all constraints (\ref{eq:EqualOptProb_F}) to (\ref{eq:EqualOptProb_theta-1})
comprise convex functions. In order to solve problem \eqref{eq:EqualOptProb},
we utilize the sequential convex programming technique which solves
instead the local linear approximation of the original problem. To
form the local linear approximation, we use the given variables $\mathcal{X}^{j},z^{j}$
in the $j$-th iteration of the algorithm to convert the above problem
to a standard convex form. For the ease of exposition, we use $c_{k}^{\text{ }}[n]$
instead of $c_{k}^{\text{ }}\left(f_{k}[n],w_{k}[n],l_{k}[n]\right)$.
First, let's start with constraint (\ref{eq:EqualOptProb_UserCap}),
since any convex functions can be lower-bounded by its first order
Taylor expansion, then we can write

\noindent 
\[
\frac{1}{N_{c}}\sum_{n=1}^{N_{c}}q_{k}[n]c_{k}^{\text{ }}[n]\ge\frac{1}{N_{c}}\sum_{n=1}^{N_{c}}q_{k}[n]\tilde{c}_{k}^{\text{ }}[n]\ge\mu_{hp},
\]

\noindent where $\tilde{c}_{k}^{\text{ }}[n]$ is an affine function
and equals to the local first order Taylor expansion of $c_{k}^{\text{ }}[n]$
and $\mu_{hp}$ is a lower bound of $\mu$. Similarly, We can convert
(\ref{eq:EqualOptProb_gm}) to (\ref{eq:EqualOptProb_theta-1}) into
the standard convex form by replacing them with their first order
Taylor expansion. We can approximate problem \eqref{eq:EqualOptProb}
as follows

\begin{subequations}\label{eq:TrjOpt_Linear}

\noindent 
\begin{align}
\max_{\mathbb{V},\mathcal{X},\mu_{hp}} & \ \mu_{hp}\\
\mbox{s.t.}\  & \frac{1}{N_{c}}\sum_{n=1}^{N_{c}}q_{k}[n]\tilde{c}_{k}^{\text{ }}[n]\ge\mu_{hp}\,,\,\forall k,\\
 & f_{k}[n]\ge\widetilde{f}_{k}[n]\ ,\,\forall k,\,\forall n,\label{eq:EqualOptProb_F-1}\\
 & w_{k}[n]\ge\widetilde{w}_{k}[n]\ ,\,\forall k,\,\forall n,\\
 & l_{k}[n]\ge\widetilde{l}_{k}[n]\ ,\,\forall k,\,\forall n,\label{eq:EqualOptProb_Dm-1}\\
 & \theta_{k}[n]\ge\widetilde{\theta}_{k}[n]\ ,\,\forall k,\,\forall n,\label{eq:EqualOptProb_rx-1}\\
 & \eqref{eq:EqualOptProb_Vge0},\,\eqref{eq:GenProb_Dmax},\,\eqref{eq:GenProb_loopconst},
\end{align}

\end{subequations}

\noindent where the superscript `` \textasciitilde{} '' denotes
the local first order Taylor expansion. Now, we have a standard convex
problem which can be solved by any convex optimization tools like
CVX\footnote{\noindent Note that to minimize the approximation error, a tight local
Taylor approximation is needed.}. We denote the generated trajectory by solving \eqref{eq:TrjOpt_Linear}
as $\mathcal{X}^{j+1}$.

\subsection{Optimal UAV Altitude\label{subsec:Opt_alt}}

\noindent Now we proceed to optimize the UAV altitude for a given
horizontal UAV trajectory $\mathcal{X}$ and scheduling decision $Q$.
Similar to the preceding section, first we introduce auxiliary variables
$h,\,m_{k}[n],\text{ and }o_{k}[n]$ consisting of convex functions
as follows
\begin{align*}
m_{k}[n] & =\exp\left(-a_{k}\arctan\left(z/r_{k}[n]\right)+b_{k}\right)\ ,\,\forall k,\,\forall n,\\
o_{k}[n] & =\left(\left(h+r_{k}^{2}[n]\right)^{(A-1)\alpha_{\ac{los}}/2}-B\right)^{-1}\ ,\,\forall k,\,\forall n,\\
h & =z^{2}.
\end{align*}

\noindent In a similar manner to section \ref{subsec:PathPlanningSec},
we find the UAV altitude by using the sequential convex programming
with given local point $z^{j}$ in the $j$-th iteration and the generated
horizontal trajectory $\mathcal{X}^{j+1}$ in the last section. Finally,
the UAV altitude is optimized as follows

\begin{subequations}\label{eq:AltOpt-Linear}

\noindent 
\begin{align}
\max_{\mathbb{W},z,\,\mu_{alt}}\  & \mu_{alt}\\
\mbox{s.t.}\  & \frac{1}{N_{c}}\sum_{n=1}^{N_{c}}q_{k}[n]\,\tilde{c}_{k}[n]\ge\mu_{alt}\ ,\,\forall k,\label{eq:EqualOptProb_UserCap-1-1}\\
\quad\  & m_{k}[n]\ge\widetilde{m}_{k}[n]\ ,\,\forall k,\,\forall n,\\
 & o_{k}[n]\ge\widetilde{o}_{k}[n]\ ,\,\forall k,\,\forall n,\\
 & h\ge\widetilde{h},\label{eq:EqualOptProb_ry-2-1}\\
 & m_{k}[n],o_{k}[n],\,h>0,\,\forall k,\,\forall n,\\
 & \eqref{eq:Alt_const},
\end{align}

\end{subequations}

\noindent where $\tilde{c}_{k}[n]$ is the first order Taylor expansion
of $c_{k}\left(m_{k}[n],o_{k}[n],h\right)$ which is a convex function
and is defined similar to (\ref{eq:CmFmDm}), and $\mathbb{W}=\left\{ m_{k}[n],o_{k}[n],h\,|\,\forall k,\,\forall n\right\} $
comprises all the auxiliary variables. The superscript `` \textasciitilde{}
'' denotes the local first order Taylor expansion, and $\mu_{alt}$
is a lower bound of $\mu$. We denote the drone altitude which is
obtained by solving \eqref{eq:AltOpt-Linear} as $z^{j+1}$ to be
used in the next iteration.

\subsection{Iterative Algorithm}

According to the preceding analysis, now we propose an iterative algorithm
to solve the original optimization problem \eqref{eq:GeneralOptProb}
by applying the block-coordinate descent method\cite{XuYin}. As mentioned
earlier, we split up our problem into three phases (or blocks) of
ground node scheduling, drone horizontal trajectory design, and flying
altitude optimization over variables $\left\{ Q,\mathcal{X},z\right\} $.
In each iteration, we update just one set of variables at a time,
rather than updating all the variables together, by fixing the other
two sets of variables. Then, the output of each phase is used as an
input for the next step. The rigorous description of this algorithm
is summarized in Algorithm \ref{alg:ItrAlg}.

\begin{algorithm}
\begin{enumerate}
\begin{spacing}{0.65}
\item Initialize all variables $\left\{ Q^{j},\mathcal{\mathcal{X}}^{j},z^{j}\right\} \,,j=1$.
\end{spacing}
\begin{spacing}{0.6}
\item Find the optimal solution of the scheduling problem \eqref{eq:ScheProb}
for given $\left\{ \mathcal{\mathcal{X}}^{j},z^{j}\right\} $. Denote
the optimal solution as $Q^{j+1}$.
\end{spacing}
\begin{spacing}{0.65}
\item Generate the optimal communication trajectory in horizontal plane
($\mathcal{\mathcal{X}}^{j+1}$) by solving \eqref{eq:TrjOpt_Linear}
with given variables $\left\{ Q^{j+1},\mathcal{\mathcal{X}}^{j},z^{j}\right\} $. 
\end{spacing}
\begin{spacing}{0.6}
\item Solving optimization problem \eqref{eq:AltOpt-Linear} given variables
$\left\{ Q^{j+1},\mathcal{\mathcal{X}}^{j+1},z^{j}\right\} $ and
denote the solution as $z^{j+1}$.
\end{spacing}
\begin{singlespace}
\item Update $j:=j+1$.
\item Go to step 2 and repeat until the convergence (i.e. until observing
a small increase in objective value).
\end{singlespace}
\end{enumerate}
\caption{Iterative algorithm for solving optimization problem \eqref{eq:GeneralOptProb}.\label{alg:ItrAlg}}
\end{algorithm}

\subsection{Proof of Convergence }

\noindent In this section we prove the convergence of Algorithm \ref{alg:ItrAlg}
in a similar manner of \cite{WuQiYoRu}. To this end, in the $j$-th
iteration we denote the $\mu(Q^{j},\mathcal{\mathcal{X}}^{j},z^{j}),\,\mu_{hp}(Q^{j},\mathcal{\mathcal{X}}^{j},z^{j}),\,\mu_{alt}(Q^{j},\mathcal{\mathcal{X}}^{j},z^{j})$
as the optimal objective values of problems \eqref{eq:ScheProb},
\eqref{eq:TrjOpt_Linear}, and \eqref{eq:AltOpt-Linear} , respectively.
From step $(2)$ of Algorithm \ref{alg:ItrAlg} for the given solution
$Q^{j+1}$, we have 
\[
\mu(Q^{j},\mathcal{\mathcal{X}}^{j},z^{j})\le\mu(Q^{j+1},\mathcal{\mathcal{X}}^{j},z^{j}),
\]

\noindent since the optimal solution of problem \eqref{eq:ScheProb}
is obtained. Moreover, we can write
\begin{align*}
\mu(Q^{j+1},\mathcal{\mathcal{X}}^{j},z^{j}) & \overset{(a)}{=}\mu_{hp}(Q^{j+1},\mathcal{\mathcal{X}}^{j},z^{j})\\
 & \overset{(b)}{\le}\mu_{hp}(Q^{j+1},\mathcal{\mathcal{X}}^{j+1},z^{j})\\
 & \overset{(c)}{\le}\mu(Q^{j+1},\mathcal{\mathcal{X}}^{j+1},z^{j}).
\end{align*}

\noindent Step $(a)$ holds due to $\mu_{hp}(Q^{j+1},\mathcal{\mathcal{X}}^{j},z^{j})$
being a tight local first order Taylor approximation of problem \eqref{eq:EqualOptProb}
at the local points. Step $(b)$ is true, since we can find the optimal
solution of problem \eqref{eq:TrjOpt_Linear} with the given variables
$\{Q^{j+1},\mathcal{\mathcal{X}}^{j},z^{j}\}$, and $(c)$ holds because
$\mu_{hp}(Q^{j+1},\mathcal{\mathcal{X}}^{j+1},z^{j})$ is the lower
bound of the objective value $\mu(Q^{j+1},\mathcal{\mathcal{X}}^{j+1},z^{j})$.
Then, by proceeding to step $(4)$ of Algorithm \ref{alg:ItrAlg}
and given variables $\{Q^{j+1},\mathcal{\mathcal{X}}^{j+1},z^{j}\}$,
we obtain
\begin{align*}
\mu(Q^{j+1},\mathcal{\mathcal{X}}^{j+1},z^{j}) & \overset{(d)}{=}\mu_{alt}(Q^{j+1},\mathcal{\mathcal{X}}^{j+1},z^{j})\\
 & \overset{(e)}{\le}\mu_{alt}(Q^{j+1},\mathcal{\mathcal{X}}^{j+1},z^{j+1})\\
 & \overset{(f)}{\le}\mu(Q^{j+1},\mathcal{\mathcal{X}}^{j+1},z^{j+1}).
\end{align*}

\noindent Step $(d)$ is true since the local first order Taylor approximation
in \eqref{eq:AltOpt-Linear} is tight for the given local variables
$\{Q^{j+1},\mathcal{\mathcal{X}}^{j+1},z^{j}\}$. $(e)$ holds since,
the optimization problem \eqref{eq:AltOpt-Linear} can be optimally
solved, and $(f)$ is true due to $\mu_{alt}(Q^{j+1},\mathcal{\mathcal{X}}^{j+1},z^{j+1})$
is a lower bound of the objective value $\mu(Q^{j+1},\mathcal{\mathcal{X}}^{j+1},z^{j+1})$.
Finally, we have
\[
\mu(Q^{j},\mathcal{\mathcal{X}}^{j},z^{j})\le\mu(Q^{j+1},\mathcal{\mathcal{X}}^{j+1},z^{j+1}).
\]

\noindent Which indicates that the objective value of Algorithm \ref{alg:ItrAlg}
after each iteration is non-decreasing and since it is upper bounded
by a finite value, so the convergence of Algorithm \ref{alg:ItrAlg}
is guaranteed.

\subsection{Trajectory Initializing\label{subsec:TrjIniti}}

In this section, we propose a simple strategy to initialize the drone
trajectory to be optimized later on by the introduced Algorithm \ref{alg:ItrAlg}.
The initial trajectory is in form of a circle which is centered at
${\bf c}_{\text{trj}}=\left(x_{\text{trj}},y_{\text{trj}}\right)$
and the radius $r_{\text{trj}}$ which is given by

\noindent 
\[
r_{\text{trj}}=\frac{L_{max}}{2\pi},
\]

\noindent where $L_{max}=T_{c}\cdot v_{max}$. To determine the ${\bf c}_{\text{trj}}$,
we use the notion of the (weighted) center of gravity of the ground
nodes \cite{EsrGesbertWSA}. Moreover, the flying altitude is initialized
at $h_{max}$.

\section{Numerical Results\label{sec:Simulation}}

We consider a dense urban Manhattan-like area of size $600\times600$
square meters, consisting of a regular street and buildings. The height
of the building is Rayleigh distributed within the range of $5$ to
$40$ meters \cite{AlhKaLa}. The average building height is $14$
m. True propagation parameters are chosen as  $\alpha_{\ac{los}}=2.27,\,\alpha_{\ac{nlos}}=3.64,\,\ss_{\ac{los}}=-30\,\text{dB},\,\ss_{\ac{nlos}}=-40\,\text{dB}$
according to an urban micro scenario in \cite{3GPP}. The variances
of the shadowing component in \ac{los} and \ac{nlos} scenarios are
$\sigma_{\ac{los}}^{2}=2$ and $\sigma_{\ac{nlos}}^{2}=5$, respectively.
The transmission power for ground nodes is chosen as $P=30$ dBm,
and the noise power at the receiver is -80 dBm. The UAV has a maximum
speed of $v_{max}=10\,\text{m}/\text{s}$.

\subsection{Learning Trajectory }

\begin{figure}[t]
\begin{centering}
\begin{tabular}{>{\centering}p{0.7cm}c}
(a) & \tabularnewline
 & \includegraphics[width=7.5cm]{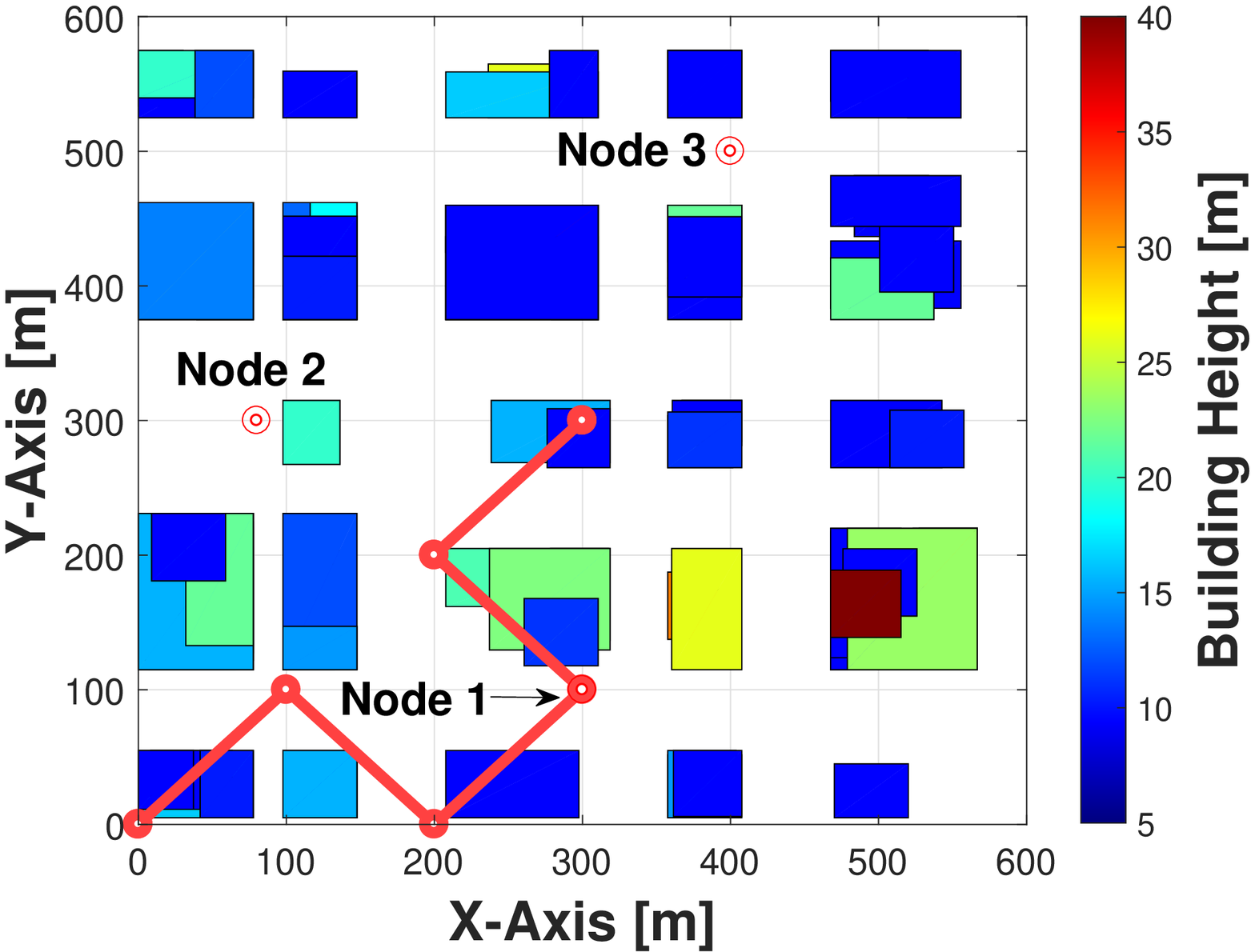}\tabularnewline
(b) & \tabularnewline
 & \includegraphics[width=7.5cm]{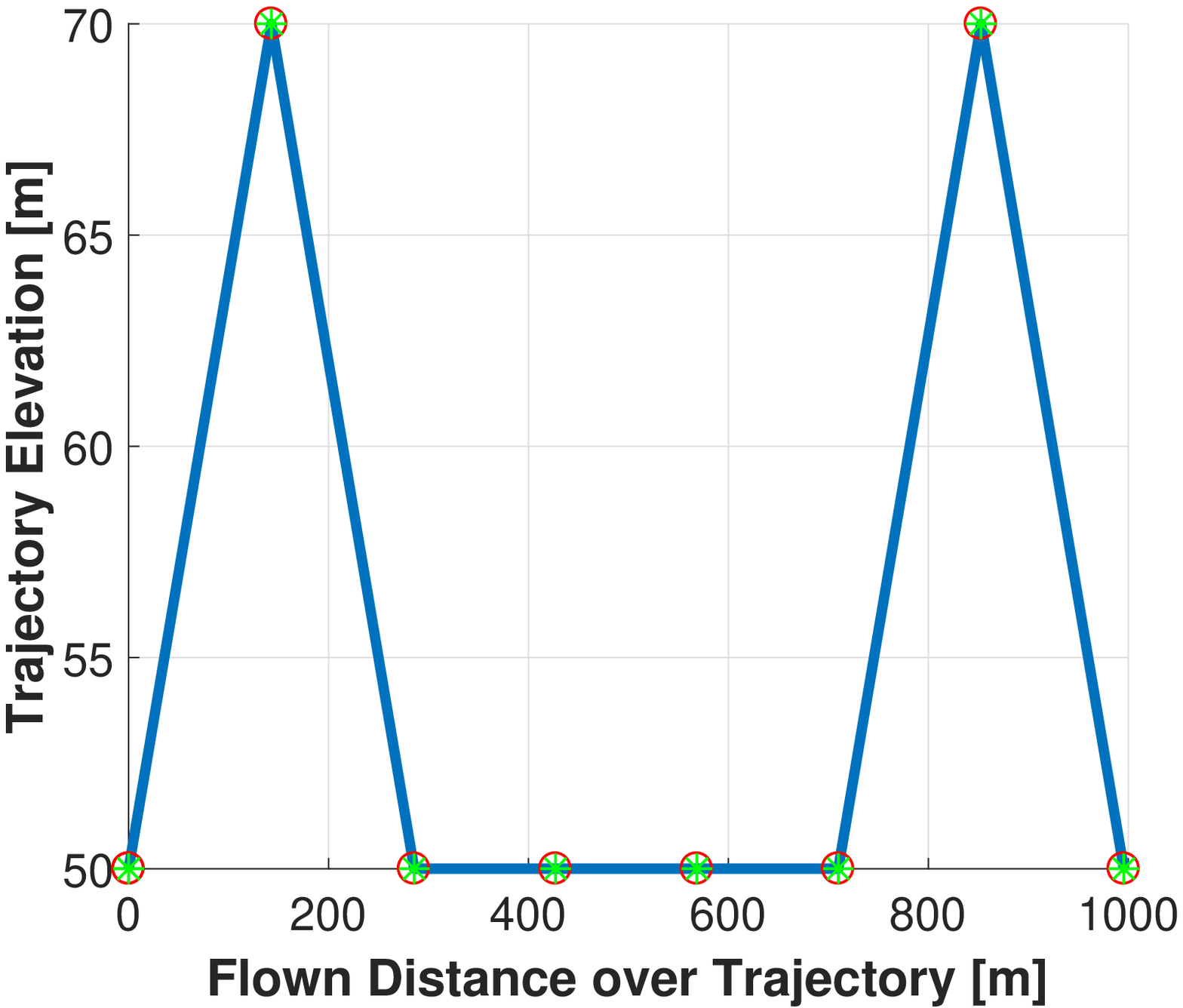}\tabularnewline
\end{tabular}
\par\end{centering}
\caption{(a) Top view of the optimal learning trajectory using proposed algorithm.
(b) The UAV elevation along the trajectory. \label{fig:LearningOptPath}}
\end{figure}

An illustration of the optimal learning trajectory is presented in
Fig. \ref{fig:LearningOptPath} for $K=3$ ground nodes. In this scenario,
the UAV flies from the base position ${\bf x}_{b}(0,0,50)$ towards
the terminal location ${\bf x}_{t}=(300,300,50)$ under the flight
time constraint $T_{l}=100\,\text{s}$. To discretize the search space
over the city for creating the 3D path graph, we chose $a_{h}=100$
m and $a_{v}=20$ m as defined in section \ref{subsec:DiscreteApproxAndDP}.
It is interesting to note that, the trajectory experiences a wide
array of altitudes there by improving the learning performance of
both \ac{los} and \ac{nlos} segments. For the ease of exhibition,
we plotted the generated trajectory in two different figures. Fig.
\ref{fig:LearningOptPath}-a shows the top view of the generated trajectory
while the elevation of the trajectory as a function of the flown distance
is depicted in Fig. \ref{fig:LearningOptPath}-b.

In Fig. \ref{fig:LearningPathPerf} the performance of the optimal
trajectory in terms of the mean squared error (MSE) of the learned
channel parameters $(\alpha_{s},\ss_{s};s=\left\{ \ac{los},\ac{nlos}\right\} )$
is shown as a function of the number of ground nodes. The duration
of the learning phase $T_{l}=100$ s. We perform Monte-Carlo simulations
over random user locations. We also compare the performance of the
optimal trajectory with that of randomly generated arbitrary trajectories.
For a given realization, arbitrary trajectory of duration $T_{l}$
is generated. It is clear that, the channel can be learned more precisely
by taking the optimized learning trajectory. Also, the learning error
is reduced when the number of ground nodes increases, since the chance
of obtaining measurement from both \ac{los} and \ac{nlos} segments
increases.

\subsection{Communication Trajectory\label{subsec:Sim_Comm_TRJ}}

\begin{figure}
\begin{centering}
\includegraphics[width=9cm]{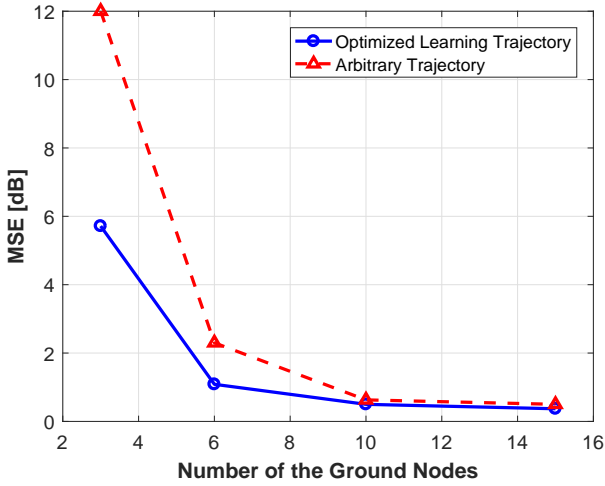}
\par\end{centering}
\caption{Comparison of the MSE for different learning trajectories.\label{fig:LearningPathPerf} }
\end{figure}

\begin{figure}
\begin{centering}
\includegraphics[width=9cm]{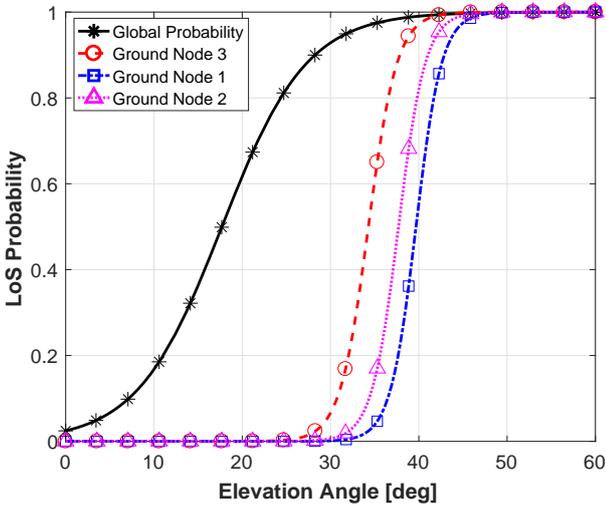}
\par\end{centering}
\caption{Global LoS probability compared with the local LoS probability learned
form the 3D map for three ground nodes. \label{fig:NodeLoSProb}}

\end{figure}

\begin{figure}
\begin{centering}
\includegraphics[width=9.1cm]{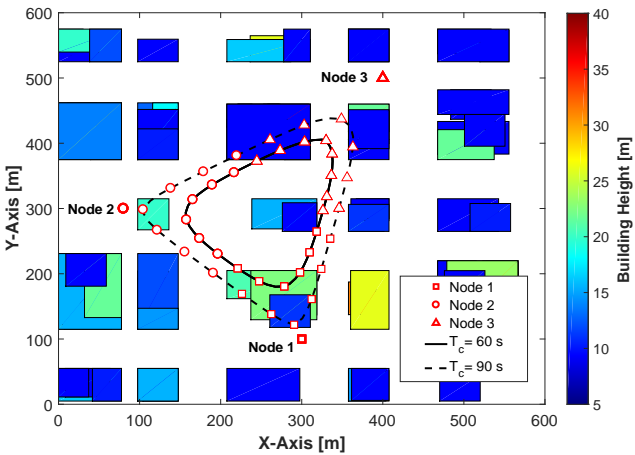}
\par\end{centering}
\caption{Optimal drone trajectory and ground node scheduling for different
flight times. As the flight time increases, the UAV gets closer to
individual ground nodes.\label{fig:DiffTrj}}
\end{figure}

In this section, we evaluate the performance of the communication
path planning algorithm. Since the communication trajectory design
depends on the local \ac{los} probability model, we first need to
learn the probability model coefficients in (\ref{eq:LoS_Prob}).
For this, we employed the logistic regression method on the training
data set obtained by randomly sampling around each ground node. The
labeling is done with the true \ac{los} status obtained from the
3D city map. Fig.\ref{fig:NodeLoSProb} shows the obtained \ac{los}
probabilities for $K=3$ ground nodes whose locations are shown in
Fig. \ref{fig:DiffTrj}. We also plot the global \ac{los} probability
which is computed from the characteristics of the 3D map according
to \cite{AlhKaLa}. It is clear that the local probabilities have
the sharper transitions and thus provide more information per ground
node (i.e. if the node is surrounded by the tall buildings or is in
a large open area), while the global probability can be considered
as the average of the local \ac{los} probability of the nodes in
different locations.

In Fig. \ref{fig:DiffTrj}, we show the generated trajectory over
the city buildings for different flight times $(T_{c})$. It is clear
that by increasing $T_{c}$, the UAV exploits the flight time to improve
the ground node link quality by enlarging the trajectory and moving
towards the ground nodes. It is crucial to note that the generated
trajectory is closer to the ground nodes which has the less \ac{los}
probability (i.e. the ground nodes who are close to buildings or surrounded
by tall skyscrapers). In Fig. \ref{fig:DiffTrj}, the drone tries
to get closer to the ground nodes 1 and 2 since they are close to
the buildings which mostly block the \ac{los} link to the drone.
Moreover, we illustrated the result of the ground node scheduling
during the trajectory with different markers which are assigned to
each node. Namely triangles, squares, and circles pertain, respectively,
to ground nodes 1 to 3. For example, square markers shown on the trajectory
indicate that the drone is serving the ground node 1 at that time.

We then outline the convergence behavior of Algorithm \ref{alg:ItrAlg}
by assuming $K=3$ and $T_{c}=90$ s. The drone altitude and worst
ground node throughput versus iteration are shown in Fig.\ref{fig:AlgEvolution}.
As we expected, the worst ground node throughput in each iteration
improves and finally converges to a finite value.

\begin{figure}
\begin{centering}
\begin{tabular}{cc}
\includegraphics[width=4.3cm]{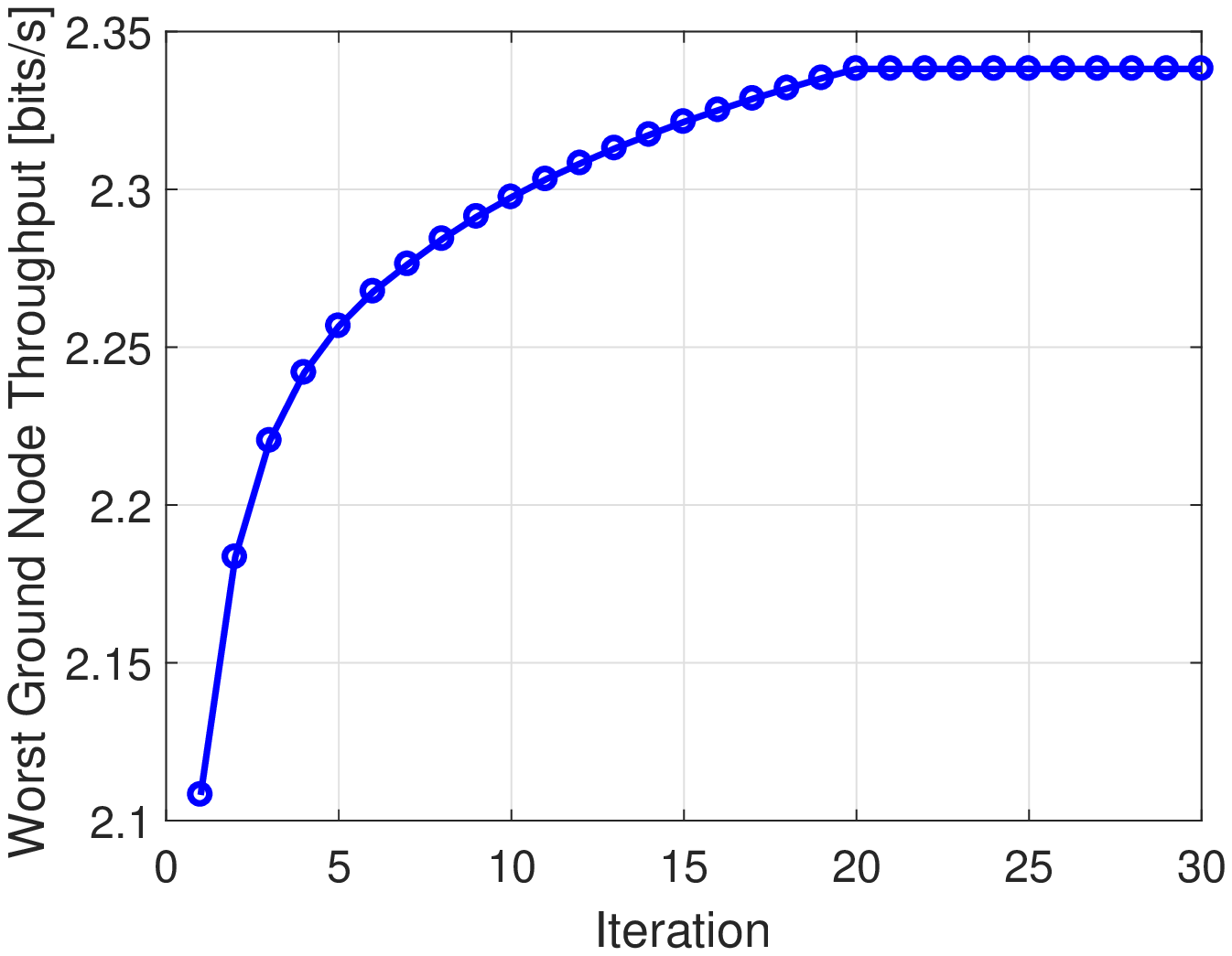} & \includegraphics[width=4.2cm]{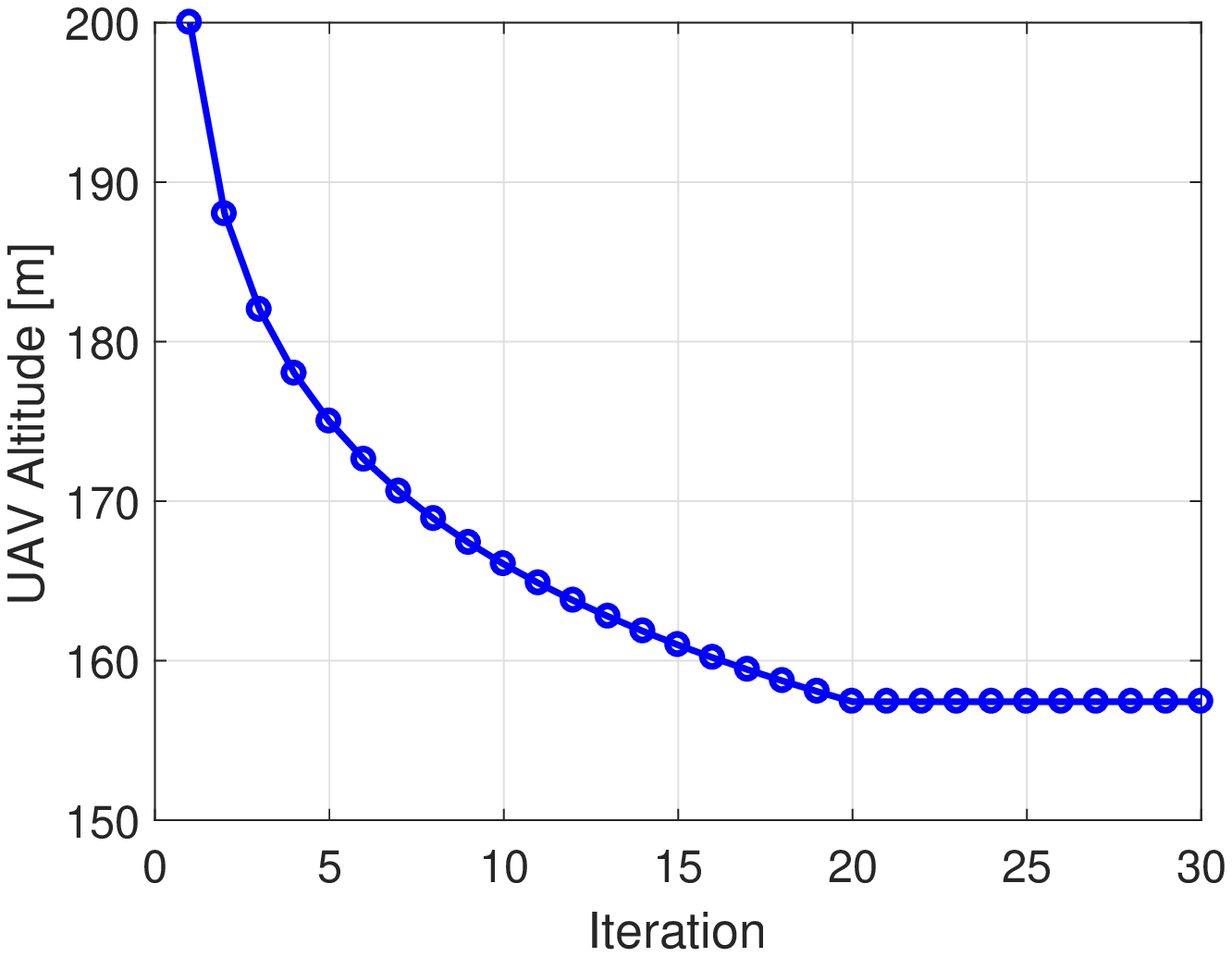}\tabularnewline
{\small{}(a)} & {\small{}(b)}\tabularnewline
\end{tabular}
\par\end{centering}
\caption{(a) Throughput performance versus iteration and (b) drone altitude
evolution versus iteration.\label{fig:AlgEvolution}}
\end{figure}

In Fig. \ref{fig:Comparison-1}, the performance of the proposed map-based
algorithm in comparison with two other approaches, which are briefly
explained below, versus the flight time by considering $K=6$ ground
nodes is shown. 
\begin{itemize}
\item \textbf{Probabilistic algorithm }

In the probabilistic approach, we consider the same trajectory design
algorithm as proposed in this paper with the difference that for a
link between the drone located at altitude $z$ and the $k$-th ground
node, the \ac{los} probability at the time step $n$ is given by

\noindent 
\[
p_{k}[n]=\frac{1}{1+\exp\left(-a\,\theta_{k}[n]+b\right)},
\]

where parameters $\left\{ a,b\right\} $ are computed according to
\cite{AlhKaLa} and based on the characteristics of the 3D map. In
other words, we use a global \ac{los} probability model.
\item \textbf{Deterministic} \textbf{algorithm}

In the deterministic algorithm, an optimal trajectory is generated
based on the method introduced in \cite{WuQiYoRu} which considers
a single deterministic \ac{los} channel model for the link between
the drone and the ground nodes. In order to have a fair comparison,
we modified this method by using an average path-loss instead of the
pure \ac{los} channel model. the channel parameters pertaining to
the average path-loss model are learned by fitting one channel model
for the whole measurements gathered from both \ac{los} and \ac{nlos}
ground nodes.
\end{itemize}
We have also investigated the impact of the imperfect estimation of
the channel parameters on the performance of the map-based algorithm.
The result of the algorithm using the learned channel parameters is
illustrated in Fig. \ref{fig:Comparison-1}. As it can be seen, the
channel estimation uncertainty stemming from the learning part has
a minor effect on the performance of the map-based algorithm and in
general the map-based algorithm outperforms the other approaches which
is expected since in the proposed algorithm we utilize more information
through the 3D map.

\begin{figure}
\begin{centering}
\includegraphics[width=9.5cm]{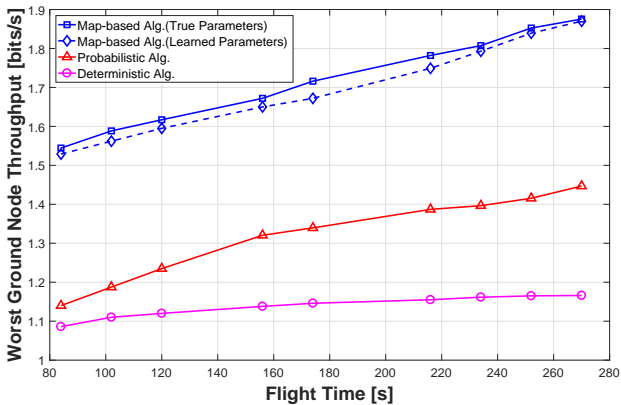}
\par\end{centering}
\caption{Performance of the map-based algorithm for the learned and true channel
parameters in comparison with the probabilistic approach and the deterministic
Algorithm for 6 ground nodes versus increasing the flight time. \label{fig:Comparison-1}}
\end{figure}

\begin{figure}
\begin{centering}
\includegraphics[width=9cm]{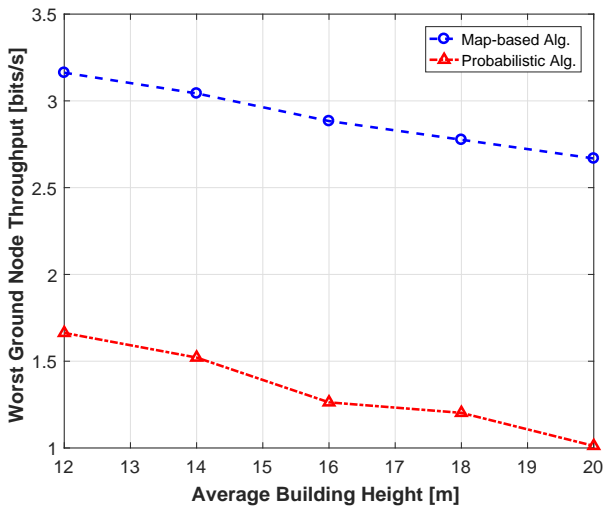}
\par\end{centering}
\caption{Performance comparison of the map-based Algorithm and the probabilistic
approach versus the average building height for a fixed flight time.\label{fig:PerVSBldHeight} }

\end{figure}

In Fig. \ref{fig:PerVSBldHeight} a performance comparison of the
proposed map-based algorithm and the probabilistic approach under
a fixed flight time by increasing the average buildings height is
illustrated. The map-based algorithm provides a better services to
the ground nodes. For both algorithms by increasing the buildings
height, the performance degraded since it is more likely that the
link between the ground nodes and the drone over the course of the
trajectory being \ac{nlos}.

\section{Conclusion\label{sec:Conclusion}}

This work has considered the problem of trajectory design for an UAV
BS that is providing communication services for a number of ground
nodes in the context of an IoT data harvesting scenario. In contrast
to the existing literature, we assume that the propagation parameters
are unknown and we devise an optimized trajectory for the UAV that
allows it to learn the parameters efficiently. The learning trajectory
optimization relies on the dynamic programming techniques and the
knowledge of the 3D city map. Once the channel parameters are learned,
we have developed throughput-optimized trajectories such that the
amount of data collected from each of the ground nodes is maximized.
We have proposed an iterative algorithm that leverages the knowledge
of the 3D city map via a novel map-compression method and uses the
block coordinate descent and sequential convex programming techniques.
It is also shown that the proposed algorithm is guaranteed to converge
to at least a locally optimal solution. The advantages of the learning
trajectory optimization and communication path planning algorithm
by utilizing the proposed map compression method are illustrated in
an urban IoT setting.

\appendix{}

\subsection{Proof of the average channel gain\label{subsec:ProofOfAvgChGain}}

The average channel gain of the link between the drone and the $k$-th
ground node in the $n$-th time slot is given by

\noindent 
\begin{align}
\text{E}[\gamma_{k}[n]] & =p_{k}[n]\gamma_{\ac{los},k}[n]+(1-p_{k}[n])\gamma_{\ac{nlos},k}[n],\label{eq:AvgChGain_Proof}
\end{align}

\noindent where $p_{s,k}[n]$ denotes the \ac{los} probability, $\gamma_{\ac{los},k}[n]$
and $\gamma_{\ac{nlos},k}[n]$, respectively, denote the channel gain
in \ac{los} and \ac{nlos} propagation segments. Expanding (\ref{eq:AvgChGain_Proof})
we have

\noindent 
\begin{align}
\text{E}[\gamma_{k}[n]]\overset{(a)}{=}\, & p_{k}[n]\frac{\beta_{\ac{los}}}{d_{k}^{\alpha_{\ac{los}}}[n]}+(1-p_{k}[n])\frac{\beta_{\ac{nlos}}}{d_{k}^{\alpha_{\ac{nlos}}}[n]}\nonumber \\
\overset{(b)}{=} & \left(\frac{d_{k}^{(A-1)\alpha_{\ac{los}}}-B}{1+\exp\left(-a_{k}\theta_{k}+b_{k}\right)}+B\right)\frac{\beta_{\ac{los}}}{d_{k}^{\alpha_{\ac{nlos}}}},\label{eq:AverageChGain_b}
\end{align}

\noindent where step $(a)$ holds by substituting the values of $\gamma_{\ac{los},k}[n]$
and $\gamma_{\ac{nlos},k}[n]$ form (\ref{eq:CH_Model}) into (\ref{eq:AvgChGain_Proof}),
$(b)$ is obtained by substituing (\ref{eq:LoS_Prob}), where $B=\frac{\beta_{\ac{nlos}}}{\beta_{\ac{los}}},\,A=\frac{\alpha_{\ac{nlos}}}{\alpha_{\ac{los}}}\ge1$,
and $d_{k}[n]=\sqrt{z^{2}+r_{k}^{2}[n]}$ is the distance between
the $k$-th ground node and the drone. Note that, in order to ease
the notation, the average random shadowing is assumed absorbed into
$\beta_{s}$ in (\ref{eq:AverageChGain_b}) i.e., $\beta_{s}\triangleq\beta_{s}exp(\sigma_{s}^{2}/2)$
, $s\in\{\ac{los},\ac{nlos}\}$.

\subsection{Proof of Lemma \ref{lem:lem1}\label{subsec:AppendA}}

By proving that the Hessian of the function $h\triangleq h(x,y)$,
is a positive semi-definite (PSD) matrix, we prove the convexity of
$h$. We start by considering the Hessian of function $\hat{h}\triangleq\hat{h}(x,y)$

\noindent 
\begin{equation}
\nabla^{2}\hat{h}=\left[\begin{array}{cc}
\frac{f_{xx}f-f_{x}^{2}}{f^{2}} & 0\\
0 & \frac{g_{yy}g-g_{y}^{2}}{g^{2}}
\end{array}\right]\ge0,\label{eq:Hess_FG}
\end{equation}

\noindent where $f\triangleq f(x)$, $g\triangleq g(x)$, $f_{x}\triangleq\frac{\partial f}{\partial x}$,
$g_{y}\triangleq\frac{\partial g}{\partial y}$, $f_{xx}\triangleq\frac{\partial^{2}f}{\partial x^{2}}$,
and $g_{yy}\triangleq\frac{\partial^{2}g}{\partial y^{2}}$. Since
$\hat{h}$ is convex, $\nabla^{2}\hat{h}$ is PSD and has non-negative
diagonal elements. Hence, for $f>0$, $g>0$,
\begin{align}
f_{xx}\ge\frac{f_{x}^{2}}{f} & ,g_{yy}\ge\frac{g_{y}^{2}}{g}.\label{eq:Hess_FG_1}
\end{align}

\noindent Also, it can easily be deduced that
\begin{equation}
f_{xx}g\ge0,\,g_{yy}f\ge0.\label{eq:fxxgge0}
\end{equation}

\noindent The hessian of $h\triangleq h(x,y)$ is given by

\noindent 
\[
\nabla^{2}h=\left[\begin{array}{cc}
\frac{g^{2}\left(f_{xx}f-f_{x}^{2}\right)+f_{xx}g}{\left(1+fg\right)^{2}} & \frac{f_{x}g_{y}}{\left(1+fg\right)^{2}}\\
\frac{f_{x}g_{y}}{\left(1+fg\right)^{2}} & \frac{f^{2}\left(g_{yy}g-g_{y}^{2}\right)+g_{yy}f}{\left(1+fg\right)^{2}}
\end{array}\right].
\]

\noindent If $\det(\nabla^{2}h)\ge0$ and $\trace(\nabla^{2}h)\ge0$,
then $\nabla^{2}h$ is PSD \cite{DeKle}. Let us rewrite $\nabla^{2}h$
as a summation of two matrices $\nabla^{2}h={\bf M}_{1}+{\bf M}_{2}$,
where

\noindent 
\begin{align*}
{\bf M}_{1} & =\left[\begin{array}{cc}
\frac{g^{2}\left(f_{xx}f-f_{x}^{2}\right)}{\left(1+fg\right)^{2}} & 0\\
0 & \frac{f^{2}\left(g_{yy}g-g_{y}^{2}\right)}{\left(1+fg\right)^{2}}
\end{array}\right],\\
{\bf M}_{2} & =\left[\begin{array}{cc}
\frac{f_{xx}g}{\left(1+fg\right)^{2}} & \frac{f_{x}g_{y}}{\left(1+fg\right)^{2}}\\
\frac{f_{x}g_{y}}{\left(1+fg\right)^{2}} & \frac{g_{yy}f}{\left(1+fg\right)^{2}}
\end{array}\right].
\end{align*}

\noindent Since $\det(\nabla^{2}h)$ is a $2\times2$ matrix, we can
write it as \cite{PreFriGar},

\noindent 
\[
\det(\nabla^{2}h)=\det({\bf M}_{1})+\det({\bf M}_{2})+\trace({\bf M}_{1}^{\dagger}{\bf M}_{2}),
\]

\noindent where ${\bf M}_{1}^{\dagger}$ is the adjugate of ${\bf M}_{1}$.
From (\ref{eq:Hess_FG}), it can easily be shown that $\det({\bf M}_{1})\ge0$.
Also, using (\ref{eq:Hess_FG_1}) we can see that
\begin{align*}
\det({\bf M}_{2}) & =\left(1+fg\right)^{-2}\left[(f_{xx}f)\,(g_{yy}g)-f_{x}^{2}g_{y}^{2}\right]\\
 & \ge0.
\end{align*}

\noindent Finally, from (\ref{eq:Hess_FG}) and (\ref{eq:fxxgge0}),
we have

\noindent 
\begin{align*}
\trace({\bf M}_{1}^{\dagger}{\bf M}_{2}) & =f^{2}\left(g_{yy}g-g_{y}^{2}\right)f_{xx}g+g^{2}\left(f_{xx}f-f_{x}^{2}\right)g_{yy}f\\
 & \ge0.
\end{align*}

\noindent Therefore, we can conclude that $\det(\nabla^{2}h)\ge0$.
It remains to prove that $\trace\left(\nabla^{2}h\right)\ge0$. Using
(\ref{eq:Hess_FG}) and (\ref{eq:fxxgge0}), we can see that the diagonal
elements of $\nabla^{2}h$ are positive and hence the $\trace\left(\nabla^{2}h\right)\ge0$.
Consequently, we can see $\nabla^{2}h$ is PSD, which concludes the
proof.

\subsection{Proof of Proposition \ref{lem:Proposition1}\label{subsec:Apx_Prop1}}

\noindent Let $f(x)=\frac{1}{1+x}$, $g(y)=\frac{1}{y}$, $h(d)=1/d^{\lambda}$
and $q(x,y)=f(x)g(y)+\tau,\,\tau\geq0$. For positive $f\triangleq f(x)$
and $g\triangleq g(y)$, since $\log(f\,g)$ is strictly convex, using
Lemma \ref{lem:lem1}, we can infer that $\log\left(q(x,y)\right)$
is also strictly convex. Finally, from the above arguments we can
easily see that the function
\[
\hat{c}(x,y,d)=\log\left(q(x,y)\,h(d)\right),\,k\geq0
\]

\noindent is also strictly convex.

The Hessian of $\hat{c}(x,y,d)$ is given by

\noindent 
\begin{equation}
\nabla^{2}\hat{c}=\left[\begin{array}{ccc}
\frac{\left(q_{xx}q-q_{x}^{2}\right)}{q^{2}} & \frac{\left(q_{xy}q-q_{x}q_{y}\right)}{q^{2}} & 0\\
\frac{\left(q_{yx}q-q_{x}q_{y}\right)}{q^{2}} & \frac{\left(q_{yy}q-q_{y}^{2}\right)}{q^{2}} & 0\\
0 & 0 & \frac{\left(h_{dd}h-h_{d}^{2}\right)}{h^{2}}
\end{array}\right],\label{eq:Hess_Log_FH}
\end{equation}

\noindent where $q\triangleq q(x,y)$ and $h\triangleq h(d)$. $q_{x},\,q_{xy}$
stand for the partial derivative of $q$ and are defined as $q_{x}=\frac{\partial q}{\partial x},\,q_{yx}=q_{xy}=\frac{\partial^{2}q}{\partial x\partial y}$.
$q_{xx},q_{yy},h_{d},h_{dd}$ also are defined similarly. Since $\nabla^{2}\hat{c}$
is a positive definite (PD) and symmetric matrix, it has positive
diagonal elements. Hence,
\begin{equation}
q_{xx}>\frac{q_{x}^{2}}{q}>0.\label{eq:q_xx}
\end{equation}

\noindent Since $h>0$, from (\ref{eq:q_xx}) we have

\noindent 
\begin{align}
 & h\,q_{xx}>0.\label{eq:HF_xx}
\end{align}

\noindent Moreover, since $\log\left(f\,g\right)$ is strictly convex,
we can write 

\noindent 
\begin{align}
f_{xx}f>f_{x}^{2},\ g_{yy}g & >g_{y}^{2}.\label{eq:F_Convex}
\end{align}

Using the above results, we now prove that the function

\noindent 
\[
c(x,y,d)=\log\left(1+q(x,y)\,h(d)\right)
\]

\noindent is convex. The Hessian of $c\triangleq c(x,y,d)$ is

\noindent 
\[
\nabla^{2}c=\frac{1}{\left(1+q\,h\right)^{2}}\left({\bf P}+{\bf Q}\right),
\]

\noindent where 
\begin{align*}
{\bf P}= & \left(q\,h\right)^{2}\nabla^{2}\hat{c},\\
{\bf Q}= & \left[\begin{array}{ccc}
q_{xx}h & q_{xy}h & q_{x}h_{d}\\
q_{yx}h & q_{yy}h & q_{y}h_{d}\\
q_{x}h_{d} & q_{y}h_{d} & q\,h_{dd}
\end{array}\right].
\end{align*}

\noindent Matrix ${\bf P}$ is PD since $\nabla^{2}\hat{c}$ is PD
and $q,h>0$. In order to show that the Hessian matrix $\nabla^{2}c$
is PD, we need to prove that ${\bf Q}$ is PD as the sum of two PD
matrices is PD. According to \cite{GilStra}, if all upper left $n\times n$
determinants of a symmetric matrix are positive, the matrix is PD.
Matrix ${\bf Q}$ is symmetric, since $q_{xy}$ and $q_{yx}$ are
equal to $f_{x}\,g_{y}$.

\noindent We start from the upper left $1\times1$ determinants of
${\bf Q}$ which equals to $q_{xx}h$. It follows from (\ref{eq:HF_xx}),
that $q_{xx}h>0$. Now, we proceed to show that the determinant of
upper left $2\times2$ matrix of ${\bf Q}$ is positive. So, we can
write
\begin{align}
\frac{\det\left({\bf Q}_{2\times2}\right)}{h^{2}} & =\left(q_{xx}\,q_{yy}-q_{xy}^{2}\right)\label{eq:Det_2by2}\\
 & \overset{\left(a\right)}{=}(f_{xx}f)(g_{xx}g)-f_{x}^{2}g_{y}^{2}\\
 & \overset{\left(b\right)}{>}0,
\end{align}
 where ${\bf Q}_{2\times2}$ denotes the upper left $2\times2$ matrix
of ${\bf Q}$, $(a)$ is obtained by substituting $q_{xx}=f_{xx}g,\,q_{yy}=g_{yy}f,\:q_{xy}=f_{x}\,g_{y}$
in (\ref{eq:Det_2by2}) and step $(b)$ follows from (\ref{eq:F_Convex}).
Then, we compute 
\begin{align*}
\det\left({\bf Q}\right) & =h_{d}^{2}\left(h\,m\right)+h_{dd}h\left(h\,q\,p\right),
\end{align*}

\noindent where $m=2q_{xy}q_{x}q_{y}-q_{xx}q_{y}^{2}-q_{yy}q_{x}^{2},\,p=q_{xx}q_{yy}-q_{xy}^{2}$.
From (\ref{eq:F_Convex}), it can be shown that $m<0$. From the convexity
of $\hat{c}(x,y,d)$ , by computing the determinant of upper left
$2\times2$ matrix of $\nabla^{2}\hat{c}$ and performing some algebraic
reductions we obtain 
\begin{align}
 & m+q\,p>0\nonumber \\
 & \Rightarrow h\,q\,p>-h\,m>0.\label{eq:Log_FG_CVX}
\end{align}

\noindent Also, since $\log\left(h\right)$ is strictly convex, we
can write
\begin{equation}
h_{dd}h>h_{d}^{2}.\label{eq:Log_H_CVX}
\end{equation}

\noindent Therefore, according to (\ref{eq:Log_FG_CVX}), (\ref{eq:Log_H_CVX}),
it can be seen that $\det\left({\bf Q}\right)>0$. Since all upper
left $n\times n$ determinants of ${\bf Q}$ are positive, we conclude
that the matrix ${\bf Q}$ is PD. Hence, $\nabla^{2}c$ is also PD.

\bibliographystyle{IEEEtran}
\bibliography{TrajectoryDesign}

\begin{IEEEbiography}[{\includegraphics[width=1in,height=1.25in,clip,keepaspectratio]{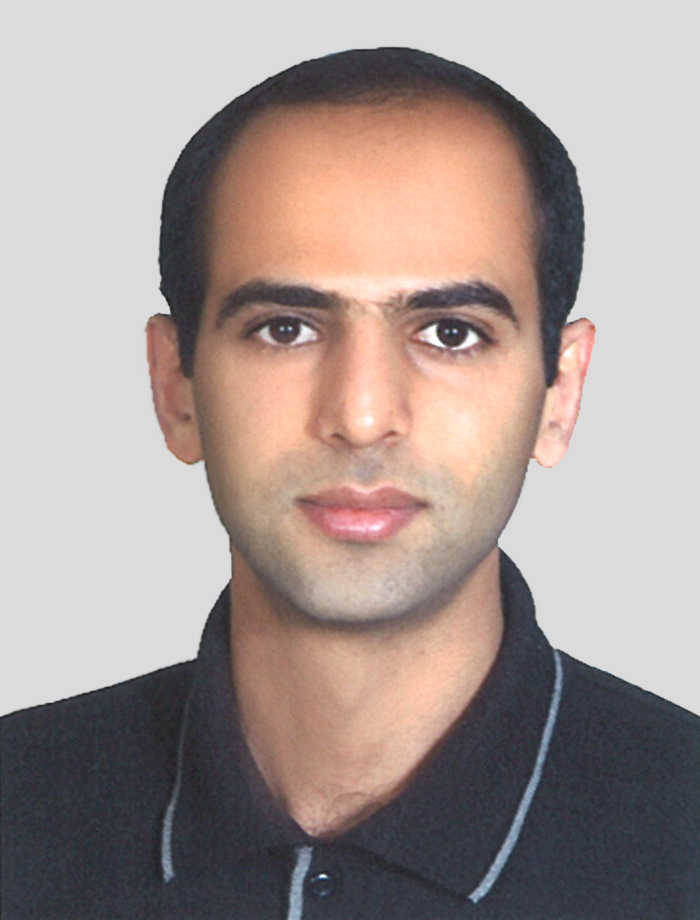}}]{Omid Esrafilian} 
received the B.Sc. degree in control and electrical engineering from K. N. Toosi University of Technology, Tehran, Iran, in 2014. He was accepted as a talented student for M.Sc. studies and he graduated as the ranked first student from K. N. Toosi University of Technology, Tehran, Iran, in 2016. He is currently pursuing the Ph.D. degree at EURECOM (Sorbonne University). He received awards from robotic competitions in the level of national and international. His main research interests include unmanned aerial vehicle communication, path planning and optimization, robotics, and control theory.
\end{IEEEbiography}

\begin{IEEEbiography}[{\includegraphics[width=1in,height=1.25in,clip,keepaspectratio]{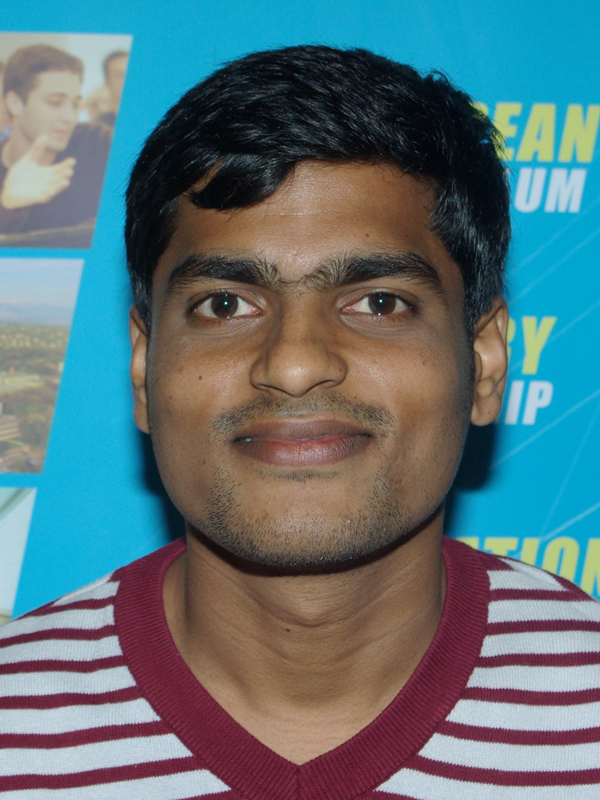}}]{Rajeev Gangula} 
 (S\textquoteright 13, M'15) received the M.Tech. degree
specializing in signal processing from Indian Institute of Technology,
Guwahati, in June 2010. He obtained the M.Sc. and Ph.D. degrees from
Telecom ParisTech (EURECOM) in 2011 and 2015, respectively. He was
awarded the Orange-MEEA scholarship for M.Sc. studies. His research
interests lie in the areas of optimization and communication theory.
\end{IEEEbiography}

\begin{IEEEbiography}[{\includegraphics[width=1in,height=1.25in,clip,keepaspectratio]{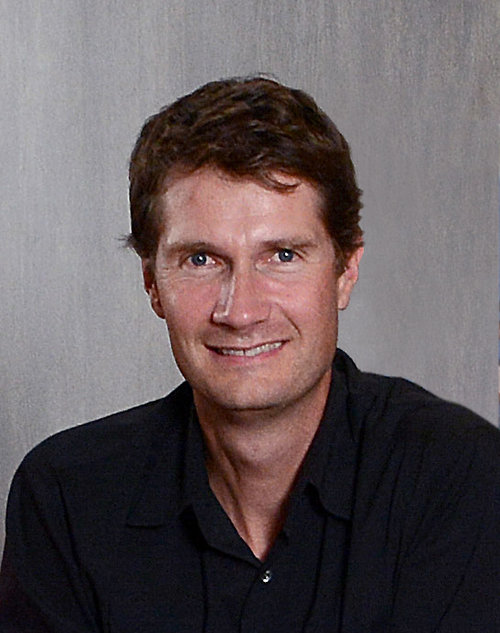}}]{David Gesbert } 
(IEEE Fellow) is Professor and Head of the Communication
Systems Department, EURECOM. He obtained the Ph.D. degree from Ecole
Nationale Superieure des Telecommunications, France, in 1997. From
1997 to 1999 he has been with the Information Systems Laboratory,
Stanford University. He was then a founding engineer of Iospan Wireless
Inc, a Stanford spin off pioneering MIMO-OFDM (now Intel). Before
joining EURECOM in 2004, he has been with the Department of Informatics,
University of Oslo as an adjunct professor. D. Gesbert has published
about 280 papers and 25 patents, some of them winning the 2015 IEEE
Best Tutorial Paper Award (Communications Society), 2012 SPS Signal
Processing Magazine Best Paper Award, 2004 IEEE Best Tutorial Paper
Award (Communications Society), 2005 Young Author Best Paper Award
for Signal Proc. Society journals, and paper awards at conferences
2011 IEEE SPAWC, 2004 ACM MSWiM. He has been a Technical Program Co-chair
for ICC2017. He was named a Thomson-Reuters Highly Cited Researchers
in Computer Science. Since 2015, he holds the ERC Advanced grant \textquotedblleft PERFUME\textquotedblright
on the topic of smart device Communications in future wireless networks.
He held visiting professor positions in KTH (2014) and TU Munich (2016).
Since 2017 he is also a visiting Academic Master within the Program
111 at the Beijing University of Posts and Telecommunications. He
sits in a Board of Directors for the OpenAirInterface Software Alliance. 
\end{IEEEbiography}

\end{document}